\documentclass[a4paper]{article}
\usepackage{amsmath,amssymb,url,graphicx}
    \addtocounter{MaxMatrixCols}{1} 
\usepackage{graphicx,booktabs,enumitem}
\usepackage{multirow}
\usepackage{tikz}
\usepackage{booktabs}
\usepackage{multirow}
\usepackage{amsthm}
    \newtheorem{aaaaa}{Do not use!}
    \newtheorem{corollary}[aaaaa]{Corollary}
    \newtheorem{definition}[aaaaa]{Definition}
    \newtheorem{example}[aaaaa]{Example}
    \newtheorem{lemma}[aaaaa]{Lemma}
    \newtheorem{remark}[aaaaa]{Remark}
    \newtheorem{proposition}[aaaaa]{Proposition}
    \newtheorem{theorem}[aaaaa]{Theorem}

\newcommand{\floor}[1]{\lfloor{#1}\rfloor}

\renewcommand{\tilde}[1]{\widetilde{#1}}
\newtheorem{Algorithm}{Algorithm}
\begin{document}

\title{Weighted Reed-Muller codes revisited}

\author{Olav Geil and Casper Thomsen\\
    \normalfont Department of Mathematical Sciences, 
    Aalborg University, Denmark\\%
    \url{olav@math.aau.dk}, 
    \url{caspert@math.aau.dk}}

\maketitle  

\begin{abstract}
We consider weighted Reed-Muller codes over point ensemble $S_1 \times
\cdots \times S_m$ where $S_i$ needs not be of the same size as
$S_j$. For $m=2$ we determine optimal weights and analyze in detail
what is the impact of the ratio $|S_1|/|S_2|$ on the minimum
distance. In conclusion the weighted Reed-Muller code construction is
much better than its reputation. For a class of affine variety codes that contains the
weighted Reed-Muller codes we then present two list decoding
algorithms. With a small modification one of these algorithms is able
to correct up to $31$ errors of the $[49,11,28]$ Joyner code.\\

\noindent\textbf{Keywords.} Affine variety codes, list decoding, weighted Reed-Muller codes 
\end{abstract}
\section{Introduction}\label{secintro}
Weighted Reed-Muller codes were introduced by S{\o}rensen
in~\cite{abs}. In his paper he demonstrates that they are subcodes of
$q$-ary Reed-Muller codes of the same minimum distance and it is
therefore not surprising that not much attention has been given to
them since. In the present paper we consider the above two code
constructions in a slightly more general setting as we allow any
point ensemble ${\mathcal{S}}=S_1\times \cdots \times S_m$, $S_1, \ldots , S_m
\subseteq {\mathbf{F}}_q$. Other authors have considered $q$-ary
Reed-Muller codes in this setting, but nobody seems to have recognized
that for such point ensembles weighted Reed-Muller codes are often
superior. We shall derive a number of results regarding their
efficiency and define what we call optimal weighted Reed-Muller
codes in two variables.\\
We argue that the dual codes are exactly as efficient and that they can be decoded up to half the
designed minimum distance by known decoding algorithms.\\
We then turn
to the decoding of weighted Reed-Muller codes. The first decoding
algorithm that we present utilizes the fact that the codes under
consideration can be viewed as subfield subcodes of certain
Reed-Solomon codes. This algorithm is a straightforward
generalization of Pellikaan and Wu's list decoding
algorithm~\cite{pw_ieee}. The second decoding algorithm that we present is
a more direct interpretation of the Guruswami-Sudan list decoding
algorithm. We are by no means the first authors to consider such an approach
for multivariate codes (see~\cite{pw_ieee}, \cite{augot},
\cite{augotstepanov}). Our contribution is that we develop a method
for deriving improved information on how many zeros of prescribed
multiplicity a multivariate polynomial can have given information about
its leading monomial with respect to the lexicographic ordering. Using
such information and allowing the decoding algorithm to perform a
preparation step we develop an improved algorithm. For some optimal
weighted Reed-Muller codes the first decoding algorithm of the paper
is quite good, for others the latter is the best.\\
Weighted Reed-Muller
codes are examples of a particular class of affine variety
codes. Whenever possible we state our findings for this more general
class of codes. As a bonus we find that when equipped with a small
trick the subfield subcode decoding algorithm can decode the Joyner
codes~\cite[Ex.\ 3.9]{joyner} beyond its minimum distance even though
till now this
code has resisted even minimum distance decoding.
{\section{A class of affine variety codes}\label{secto}}
Given
$${\mathcal{S}}=S_1 \times \cdots \times S_m=\{P_1, \ldots ,P_{|{\mathcal{S}}|}\}$$ 
write $n=|{\mathcal{S}}|$ and consider the evaluation map 
$${\mbox{ev}}_{\mathcal{S}}:{\mathbf{F}}_q[X_1,
\ldots ,X_m] \rightarrow {\mathbf{F}}_q^{n}, {\mbox{ \ \ \ \ }}{\mbox{ev}}_{\mathcal{S}}(F)=(F(P_1),
\ldots , F(P_n)).$$ 
Let 
$${\mathbb{M}}\subseteq \{X_1^{i_1} \cdots X_m^{i_m} \mid 0 \leq i_j <
|S_j|, j=1, \ldots , m\}$$
and define the affine variety code
$$E({\mathbb{M}},{\mathcal{S}})={\mbox{Span}}_{{\mathbf{F}}_q}\{{\mbox{ev}}_{\mathcal{S}}(M)
\mid M \in {\mathbb{M}} \}.$$
Throughout the paper we use the notation
$s_i=|S_i|$ for $i=1, \ldots , m$. If not explicitly stated we shall always 
assume that the enumeration is made such that $s_1 \geq \cdots \geq s_m$ holds. In the special case
that $S_1=\cdots=S_m$ we write ${\mathcal{S}}=S \times
\cdots \times S$ and $s=|S|$. 
We first show
how to find the dimension of the code. 
\begin{proposition}
The dimension of $E({\mathbb{M}},{\mathcal{S}})$ equals $|{\mathbb{M}}|$.
\end{proposition}
\begin{proof}
We only need to show that 
$$\{ {\mbox{ev}}_{\mathcal{S}}(X_1^{i_1}, \ldots , X_m^{i_m}) \mid 
 0 \leq i_j <
s_j, j=1, \ldots , m\}$$
constitutes a basis for ${\mathbf{F}}_q^{n}$ as a vectorspace over
${\mathbf{F}}_q$. For this purpose it is sufficient to show that the 
restriction of ${\mbox{ev}}_{\mathcal{S}}$ to 
\begin{eqnarray}
\{ G(X_1, \ldots , X_m) \mid \deg_{X_i} G < s_i, i=1, \ldots ,m\} \label{equi}
\end{eqnarray}
is surjective. Given $(a_1, \ldots ,a_{n})\in {\mathbf{F}}_q^{n}$ let
$$F(X_1, \ldots , X_m)=\sum_{v=1}^n a_v \prod_{i=1}^m \prod_{a \in {\mathbf{F}}_q
  \backslash \{P_i^{(v)}\}}\bigg( \frac{X_i-a}{P_i^{(v)}-a} \bigg).$$
Here, we have used the notation $P_v=(P_1^{(v)}, \ldots , P_n^{(v)})$,
$v=1, \ldots , n$. 
It is clear that ${\mbox{ev}}_{\mathcal{S}}(F)=(a_1, \ldots , a_n)$ and therefore
${\mbox{ev}}_{\mathcal{S}} : {\mathbf{F}}_q [X_1, \ldots , X_m] \rightarrow
{\mathbf{F}}_q^n$ is surjective. Consider an arbitrary monomial
ordering. Let $R(X_1, \ldots , X_m)$ be the remainder of $F(X_1,
\ldots ,X_m)$ after division with
$$\{ \prod_{a \in S_1}(X_1-a), \ldots ,\prod_{a \in S_m}(X_m-a)\}.$$ 
Clearly, $F(P_i)=R(P_i)=a_i$, $i=1, \ldots , n$. Hence, the
restriction of ${\mbox{ev}}_{\mathcal{S}}$ to~(\ref{equi}) is indeed surjective.
\end{proof}
We next show how to estimate the minimum distance of
$E({\mathbb{M}},{\mathcal{S}})$. The Schwartz-Zippel bound \cite{Schwartz,Zippel,DeMilloLipton} is as follows:
\begin{theorem}\label{thorgsz}
Given a lexicographic ordering let the leading monomial of $F(X_1,
\ldots , X_m)$ be $X_1^{i_1}\cdots X_m^{i_m}$.
The number of elements in
${\mathcal{S}}=S_1 \times \cdots \times S_m$ that are zeros of $F$ is at most
equal to 
$$i_1s_2\cdots s_m+s_1 i_2 s_3 \cdots s_m+\cdots
+s_1\cdots s_{m-1} i_m.$$
\end{theorem}
The proof of this result is purely combinatorial. Using the
inclusion-exclusion principle it can actually be strengthened to the
following result which is a special case of the footprint bound from Gr\"{o}bner basis
theory: 
\begin{theorem}\label{footprspecial}
Given a lexicographic ordering let the leading monomial of $F(X_1,
\ldots , X_m)$ be $X_1^{i_1}\cdots X_m^{i_m}$. The number of elements in
${\mathcal{S}}=S_1 \times \cdots \times S_m$ that are zeros of $F$ is at most
equal to $$ n-(s_1-i_1)(s_2-i_2)\cdots (s_m-i_m). $$
\end{theorem}
\begin{proposition}\label{afstand}
The minimum distance of $E({\mathbb{M}},{\mathcal{S}})$ is at least 
$$\min \{ (s_1-i_1)(s_2-i_2)\cdots (s_m-i_m) | X_1^{i_1}\cdots
X_m^{i_m} \in {\mathbb{M}}\}.$$
The bound is sharp if for every $M\in {\mathbb{M}}$ all divisors of $M$ also belong to
${\mathbb{M}}$.
\end{proposition}
\begin{proof}
The first part follows from Theorem~\ref{footprspecial}. To see the
last part write for $i=1, \ldots , m$, $S_i=\{b_1^{(i)}, \ldots , b_{|S_i|}^{(i)}\}$.
The polynomial 
$$F(X_1, \ldots  X_m)=\prod_{v=1}^m \prod_{j=1}^{i_v} \big(X_v-b_j^{(v)}\big)$$
has leading monomial $X_1^{i_1}\cdots X_m^{i_m}$ with respect to any
monomial ordering and evaluates to zero in exactly $n- (s_1-i_1)(s_2-i_2)\cdots (s_m-i_m)$ points from ${\mathcal{S}}$. Finally, any monomial that occurs in the support
of $F$ is a factor of $X_1^{i_1} \cdots X_m^{i_m}$.
\end{proof}
\section{Weighted Reed-Muller codes}\label{secwrm}
The first example of codes $E({\mathbb{M}},{\mathcal{S}})$ that comes to mind are the $q$-ary Reed-Muller codes ${\mbox{RM}}_q(u,m)$. They
are defined by choosing
\begin{eqnarray}
&S_1=\cdots =S_m={\mathbf{F}}_q, \nonumber\\
&{\mathbb{M}}=\{X_1^{i_1} \cdots X_m^{i_m} \mid i_1+\cdots +i_m\leq
u\}. \label{dabel1}
\end{eqnarray}
S{\o}rensen in~\cite{abs} modified the above construction by instead letting
\begin{eqnarray}
&{\mathbb{M}}=\{X_1^{i_1} \cdots X_m^{i_m} \mid w_1 i_1+\cdots +w_m i_m\leq
u\} \label{dabel2}
\end{eqnarray}
where $w_1, \ldots ,w_m$ are fixed positive numbers. The resulting
codes are called weighted Reed-Muller codes. In the same paper
S{\o}rensen argues that there is actually no point in
considering~(\ref{dabel2}) rather than~(\ref{dabel1}) as every
weighted Reed-Muller code is contained in a code ${\mbox{RM}}_q(u,m)$
which has the same minimum distance. In the present paper we allow
$S_1, \ldots , S_m$ to be any subsets of ${\mathbf{F}}_q$. As we shall
demonstrate, in such a general setting replacing~(\ref{dabel1})
with~(\ref{dabel2}) may result in much better codes. In other words, 
the concept of weighted Reed-Muller codes actually makes a lot of
sense. We start with a motivating example.
\begin{example}
In this example we construct codes over ${\mathbf{F}}_{16}$ of length
$n=64$. First let ${\mathcal{S}}=S_1\times S_2$ be such that $s_1=s_2=8$. Define,
\begin{eqnarray}
&{\mathbb{M}}=\{X_1^{i_1} X_2^{i_2} \mid 0\leq i_1, i_2\leq 7, i_1+i_2\leq 7 \}. \nonumber
\end{eqnarray}
The code $E({\mathbb{M}},{\mathcal{S}})$ is of dimension $36$ and minimum distance
$8$. Letting instead $\tilde{{\mathcal{S}}}=\tilde{S}_1\times \tilde{S}_2$ where
$|\tilde{S}_1|=16$ and $|\tilde{S}_2|=4$ we consider the following two
sets of monomials
\begin{eqnarray}
&{\mathbb{M}}'=\{X_1^{i_1}X_2^{i_2} \mid 0 \leq i_1 \leq 15, 0 \leq i_2
\leq 3, i_1+i_2\leq 11\}, \nonumber\\
&{\mathbb{M}}''=\{X_1^{i_1}X_2^{i_2} \mid 0 \leq i_1 \leq 15, 0\leq i_2
\leq 3, i_1+2i_2\leq 14\}. \nonumber
\end{eqnarray}
The code $E({\mathbb{M}}',\tilde{{\mathcal{S}}})$ is of dimension $42$ and minimum
distance $8$ whereas the code   $E({\mathbb{M}}'',\tilde{{\mathcal{S}}})$ is of
dimension $48$ and minimum distance $8$. 
\end{example}
The above example illustrates two facts. Firstly, choosing the $S_i$'s
to be of different sizes may be an advantage. Secondly, using a
weighted degree rather than the total degree when choosing monomials
may result in better codes. It is time for a definition.
\begin{definition}
Let $S_1, \ldots , S_m \subseteq {\mathbb{F}}_q$ and consider positive
numbers $w_1, \ldots , w_m,u$. Let 
\begin{eqnarray}
\hspace*{-1em}&{\mathbb{M}}=\{X_1^{i_1} \cdots X_m^{i_m} \mid 0\leq i_t\leq s_t-1,
t=1, \ldots , m, {\mbox{\ and \ }} w_1i_1+\cdots +w_mi_m\leq u\}. \nonumber
\end{eqnarray}
The corresponding code $E({\mathbb{M}},{\mathcal{S}})$ is called a weighted
Reed-Muller code and we denote it by ${\mbox{RM}}(S_1, \ldots ,
S_m,u,w_1, \ldots , w_m)$. As is often done we shall refer to weighted
Reed-Muller codes with  $S_1=\cdots
=S_m$ and $w_1=\cdots =w_m$ as $q$-ary Reed-Muller codes.
\end{definition}
We start by taking a closer look at the case of two
variables. According to Theorem~\ref{afstand} the minimum distance of
${\mbox{RM}}(S_1,S_2,u,w_1,w_2)$ equals
\begin{multline}
\hspace*{-1em}\min \{ (s_1-i_1)(s_2-i_2) \mid i_1, i_2 \in {\mathbb{N}}, 0\leq
i_1\leq s_1-1 ,
0\leq i_2\leq s_2-1, w_1i_1+w_2i_2\leq
u\} \\
\geq\min \{ (s_1-i_1)(s_2-i_2) \mid i_1, i_2 \in {\mathbb{Q}}, 0\leq
i_1\leq s_1-1,\\ 0\leq i_2\leq s_2-1, w_1i_1+w_2i_2=
u\}. \label{eqkloer}
\end{multline}
Substituting $i_2=(u-w_1i_1)/w_2$ into $(s_1-i_1)(s_2-i_2)$ we get a
concave function (a parabola). Hence, the minimal value of
$(s_1-i_1)(s_2-i_2)$ under the condition in~(\ref{eqkloer}) is either
attained for $i_1$ as small as possible or for $i_1$ as large as
possible. Given a weight $w_1$ and a positive number $u$ we seek $w_2$
such that $(s_1-i_1)(s_2-i_2)$ is the same for $i_1$ minimal and
maximal under the condition in~(\ref{eqkloer}). 
\begin{proposition}\label{optimal}
Let $s_2\leq s_1$ be positive integers. Given fixed positive numbers
$w_1$ and $u$ assume $w_2$ is chosen to be the positive number such that
$(s_1-i_1)(s_2-i_2)$ attains the same value whenever $i_1$ is
minimal or is maximal under the condition
\begin{eqnarray}
&w_1i_1+w_2i_2=u, \nonumber\\
&0\leq i_1\leq s_1-1, {\mbox{ \ \ }} 0\leq i_2\leq s_2-1.\nonumber
\end{eqnarray}
We have
\begin{equation}
\frac{w_1}{w_2}=\left\{ \begin{array}{cl}
s_2/s_1&{\mbox{ \ if \ }} 0 < u \leq
(s_1-\frac{s_1}{s_2})w_1\\
w_1/(w_1s_1-u)& {\mbox{ \ if \ }} (s_1-\frac{s_1}{s_2})w_1\leq u
\leq (s_1-1)w_1 \\
1& {\mbox{ \ if \ }} (s_1-1)w_1 \leq u < (s_1-1)w_1+(s_2-1)w_2.
\end{array}\right. \label{eqnogetsomething}
\end{equation}
\end{proposition}
\begin{proof}
The proposition is illustrated in Figure~\ref{figo} for the case of
$s_1=18$ and $s_2=6$. \\
\begin{figure}
    \centering
    \begin{tikzpicture}
        [x=1.5em,y=1.5em,>=latex,
        graydot/.style={circle,fill=gray,inner sep=0pt,minimum size=.375em},
        blackdot/.style={circle,fill=black,inner sep=0pt,minimum size=.525em}]
        \draw [->,very thick] (0,0) node {} -- (18,0);
        \draw [->,very thick] (0,0) node {} -- (0,6);
        \draw (17,0) node {} -- (17,5);
        \draw (0,5) node {} -- (17,5);
        \foreach \x in {0,...,17}
            \foreach \y in {0,...,5}
                \draw (\x,\y) node [graydot] {};
        \draw (0,5) node [blackdot] {};
        \node at (0,5) [anchor=east] {$s_2-1$};
        \path[use as bounding box] (-2,0) rectangle ++(21,00);
        \draw (15,0) node [blackdot] {};
        \node at (15,-.75) [left=-1em] {$s_1-\frac{s_1}{s_2}$};
        \draw (17,0) node [blackdot] {};
        \node at (17,-.75) [right=-1em] {$s_1-1$};
        \draw (3,0) -- (0,1);
        \draw (6,0) -- (0,2);
        \draw (9,0) -- (0,3);
        \draw (12,0) -- (0,4);
        \draw [line width=1.5pt] (15,0) node {} -- (0,5);
        \draw (15.5,0) node {} -- (3,5);
        \draw (16,0) node {} -- (6,5);
        \draw (16.5,0) node {} -- (9,5);
        \draw [line width=1.5pt] (17,0) node {} -- (12,5);
        \draw (17,1) -- (13,5);
        \draw (17,2) -- (14,5);
        \draw (17,3) -- (15,5);
        \draw (17,4) -- (16,5);
    \end{tikzpicture}
    \caption{The situation in the proof of Proposition~\ref{optimal}.}
    \label{figo}
\end{figure}
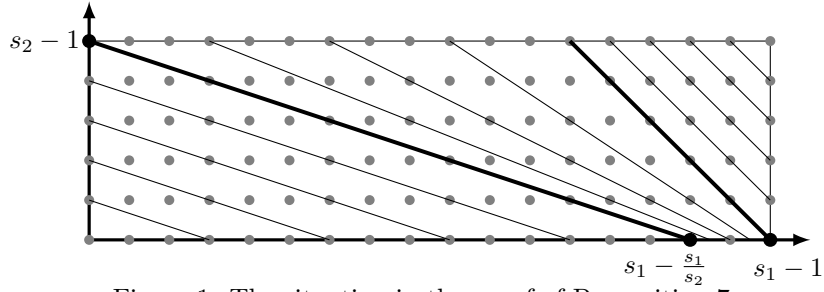
We concentrate on the situation where 
$$
(s_1-\frac{s_1}{s_2})w_1 \leq u \leq (s_1-1)w_1
$$
and leave the other two simpler cases for the reader. Write
$u/w_1=s_1-\delta$ with $s_1/s_2 \geq \delta \geq 1$. The maximal
value of $i_1$ is $u/w_1$ in which case $i_2=0$. So for $i_1$ maximal
$(s_1-i_1)(s_2-i_2)=\delta s_2$. We seek $i_1$ minimal such that with
$i_2=s_2-1$ we get $(s_1-i_1)(s_2-i_2)=\delta s_2$. We find
$i_1=s_1-\delta s_2$ which is indeed a non-negative number. Hence,
$w_2$ must satisfy
\begin{eqnarray}
&&w_1(s_1-\delta s_2)+w_2(s_2-1)=w_1(s_1-\delta)\nonumber \\
&\Downarrow & \nonumber \\
&&\frac{w_1}{w_2}=\frac{1}{\delta} \nonumber \\
&\Downarrow \nonumber \\
&&\frac{w_1}{w_2}=\frac{w_1}{w_1s_1-u}. \nonumber
\end{eqnarray}
\end{proof}
Proposition~\ref{optimal}
justifies the
following definition.
\begin{definition}\label{defoptimal}
If $s_1, s_2, u, w_1, w_2$ satisfy~(\ref{eqnogetsomething}) then
the code ${\mbox{RM}}(S_1,S_2,u,w_1,w_2)$ is called an optimal
weighted Reed-Muller code (in two variables).
\end{definition}
The next proposition estimates the minimum distance of any
weighted Reed-Muller code ${\mbox{RM}}(S_1,S_2,u,w_1,w_2)$ (optimal or
not).
\begin{proposition}\label{wrmafstand}
Consider ${\mbox{RM}}(S_1,S_2,u,w_1,w_2)$ with $s_2 \leq s_1$. Write $\rho=w_1/w_2$ and
let $d$ be the minimum distance.\\
If $\rho \leq \frac{s_2}{s_1}$ then
\begin{align}
d&\geq s_2(s_1-\frac{u}{w_1}), && \textrm{if}\ u \leq (s_1-1)w_1,
\label{EQ1}\\
d&\geq s_2-\frac{u-(s_1-1)w_1}{w_2},&& \textrm{if}\ (s_1-1)w_1 < u
\leq (s_1-1)w_1+(s_2-1)w_2. \label{EQ2}
\intertext{If $\frac{s_2}{s_1} < \rho <1$ then}
d&\geq
(s_2-\frac{u}{w_2})s_1,&& \textrm{if}\ u \leq (s_2-1)w_2, \label{EQ3}\\
d&\geq s_1-\frac{u-(s_2-1)w_2}{w_1},&&\textrm{if}\ (s_2-1)w_2 < u
\leq (s_1-\frac{1}{\rho})w_1,\label{EQ4}\\
d&\geq (s_1-\frac{u}{w_1})s_2,&& \textrm{if}\
(s_1-\frac{1}{\rho})w_1 < u \leq (s_1-1)w_1,\label{EQ5}\\
d&\geq s_2-\frac{u-(s_1-1)w_1}{w_2},&&\textrm{if}\ (s_1-1)w_1 < u
\leq (s_1-1)w_1+(s_2-1)w_2 .\label{EQ6}
\intertext{If $1  \leq \rho$ then}
d&\geq (s_2-\frac{u}{w_2})s_1,&&\textrm{if}\ u\leq (s_2-1)w_2, \label{EQ7}\\
d&\geq s_1-\frac{u-(s_2-1)w_2}{w_1},&&\textrm{if}\
(s_2-1)w_2<u\leq(s_1-1)w_1+(s_2-1)w_2.\label{EQ8}
\end{align}
Equality holds in~(\ref{EQ2}), (\ref{EQ4}), (\ref{EQ6}), and
(\ref{EQ8}), respectively, if the expression is an integer. Equality
holds in (\ref{EQ1}) and  (\ref{EQ5}) if $u/w_1$ is an
integer. Finally, equality holds in (\ref{EQ3}) and (\ref{EQ7}) if
$u/w_2$ is an integer.
\end{proposition}
\begin{proof}
The task is to determine under the various conditions of the
proposition whether $(s_1-i_1)(s_1-i_2)$ is minimized for $i_1$
minimal or maximal. The corresponding values of $i_1$ and $i_2$ are
then plugged in to give (\ref{EQ1}),$\cdots$,(\ref{EQ8}). To find
out if $i_1$ should be chosen minimal or maximal we use the
information from Proposition~\ref{optimal}. If $\rho \leq s_2/s_1$ the
minimum is always attained for $i_1$ maximal. If $1 \leq \rho$ then
the minimum is always attained for $i_1$ minimal. In the case $s_2/s_1
<\rho<1$ the minimal is attained for $i_1$ minimal when $u \leq u'$
and is attained for $i_1$ maximal when $u \geq u'$. Here, $u'$ is a
number that we determine below. It is clear that 
$$(s_1-\frac{s_1}{s_2})w_1<u'<(s_1-1)w_1$$
and therefore $u'$ is the number such that
$$(s_1-\frac{u'}{w_1})s_2=s_1-\frac{u' -(s_2-1)w_2}{w_1}.$$
Solving for $u'$ gives
$$u'=w_1s_1-w_2=(s_1-\frac{1}{\rho})w_1.$$ 
\end{proof}
Proposition~\ref{wrmafstand} also allows us to state general bounds
for the minimum distance of ${\mbox{RM}}(S_1, \ldots ,S_m,u,w_1,
\ldots ,w_m)$ in two important cases. Observe, that in particular the
following proposition can be applied when $w_i=\prod_{i \neq j}s_j$.
\begin{proposition}
Assume $s_1 \geq \cdots \geq s_m$, and let $u$ be a number $0\leq u
\leq (s_1-1)w_1+\cdots +(s_m-1)w_m$. If
\begin{equation}
\frac{w_1}{\prod_{i\neq 1}s_i} \leq \frac{w_2}{\prod_{i\neq 2}s_i}\leq
\cdots \leq \frac{w_m}{\prod_{i\neq m}s_i} \label{cirkela}
\end{equation}
holds then write 
$$u=(s_1-1)w_1 +\cdots +(s_{t-1}-1)w_{t-1}+a_tw_t$$
where $0<a_t\leq s_t-1$. The minimum distance of  ${\mbox{RM}}(S_1, \ldots ,S_m,u,w_1,
\ldots ,w_m)$ satisfies
$$d\geq (s_t-a_t)\prod_{i=t+1}^m s_i$$
with equality if $a_{t}$ is an integer.\\
If $w_1 \geq \cdots \geq w_m$ then write
$$u=(s_m-1)w_m +\cdots +(s_{t-1}-1)w_{t-1}+a_tw_t$$
where $0<a_t\leq s_t-1$. The minimum distance of  ${\mbox{RM}}(S_1, \ldots ,S_m,u,w_1,
\ldots ,w_m)$ satisfies
$$d\geq (s_t-a_t)\prod_{i=1}^{t-1} s_i$$
with equality if $a_{t}$ is an integer.\\
\end{proposition}
\begin{proof}
We only prove the first part. Assume~(\ref{cirkela}) holds. Let $i_1,
\ldots ,i_m\in {\mathbb{Q}}$ be chosen such that $(s_1-i_1)\cdots
(s_m-i_m)$ is minimal under the conditions
$$w_1i_1+\cdots +w_mi_m=u,$$
$$0\leq i_1 \leq s_1-1, \ldots , 0\leq i_m\leq s_m-1.$$
For integers $c,d$ with $1 \leq c<d\leq m$ we have $w_c/w_d \leq
s_d/s_c$. Note from Proposition~\ref{optimal} that $s_d/s_c$ is the
smallest possible ratio of $w_c^\prime/w_d^\prime$ for an optimal code
${\mbox{RM}}(S_c,S_d,i_cw_c+i_dw_d,w_c^\prime,w_d^\prime)$. Therefore,
under the condition that $i_cw_c+i_dw_d$ is fixed and $0 \leq i_c \leq
s_c-1$, $0\leq i_d \leq s_d-1$ the minimal value of
$(s_c-i_c)(s_d-i_d)$ is attained for $i_d$ minimal. The result now
follows by induction.
\end{proof}
In the remaining part of this section we restrict solely to the case
of two variables. As shall be demonstrated in this situation almost
all weighted Reed-Muller codes outperform the corresponding $q$-ary
Reed-Muller codes. Before getting to the analysis let us consider an example.
\begin{example}\label{exto}
Consider optimal weighted Reed-Muller codes
${\mbox{RM}}(S_1,S_2,u,w_1,w_2)$ (see
Definition~\ref{defoptimal}). Choosing $(s_1,s_2)$ from the set
$$ \{(32,32),(64,16),(128,8),(256,4),(512,2)\}$$
gives five different classes of codes all of length
$n=1024$. Observe that the first class of codes is similar to
$q$-ary Reed-Muller codes as the optimal choice of $w_1, w_2$ is
$w_1=w_2$ whenever $s_1=s_2$. The codes are defined whenever the field
under consideration contains at least $s_1$ elements. Hence, the first
class of codes is defined over any field ${\mathbf{F}}_q$ with $q
\geq 32$, the second class over any field ${\mathbf{F}}_q$ with $q
\geq 64$, ..., the last class of codes over any field ${\mathbf{F}}_q$
with $q \geq 512$. In particular all classes of codes are defined over
$F_{512}$. In Figure~\ref{figtvo} we compare their
performance. It is clear that the second class of codes outperforms
the first class for higher dimensions, whereas the last three classes
of codes outperform the first class for any dimension. 
\begin{figure}
\begin{center}
\includegraphics[width=8cm]{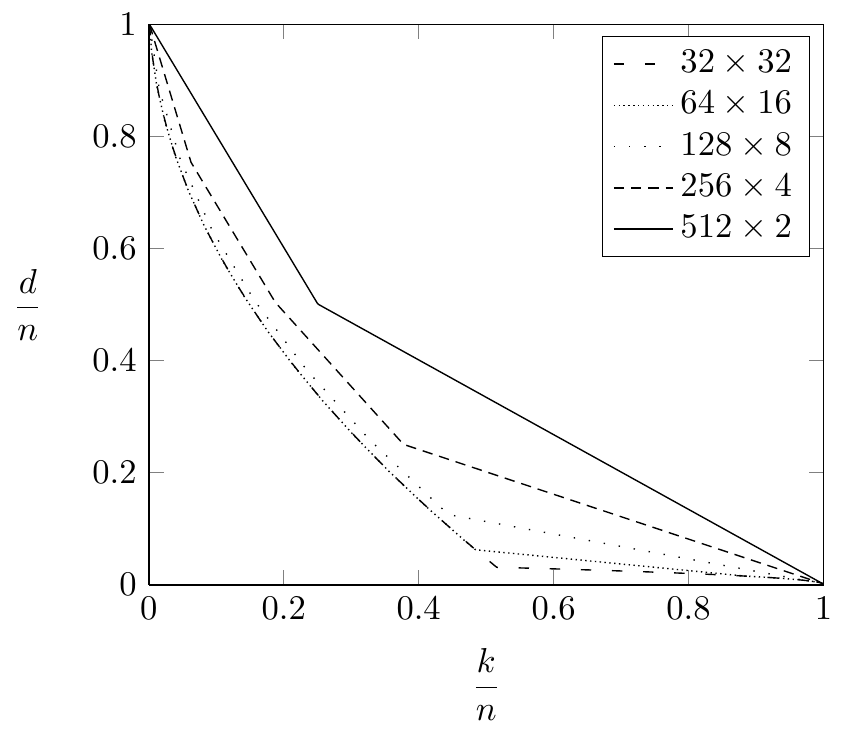} 
\end{center}
\caption{Performance of the codes in Example~\ref{exto}}
\end{figure}
\label{figtvo}
\end{example}
Below we investigate in detail how well general optimal
weighted Reed-Muller codes ${\mbox{RM}}(S_1,S_2,u,w_1,w_2)$ perform in
comparison with $q$-ary Reed-Muller codes
${\mbox{RM}}(S,S,u^\prime,1,1)$. Here, we assume that
$s_1s_2=s^2$. Recall, from Proposition~\ref{optimal} that the
description of the weights used in the optimal weighted Reed-Muller
codes involves three cases depending on the value of $u$. Choosing in
the following without loss of generality $w_1=1$ we shall refer to $u
\leq s_1-(s_1/s_2)$ as region I, $s_1-(s_1/s_2) \leq u \leq s_1-1$ as
region II, and finally $s_1-1 \leq u \leq (s_1-1)+w_2(s_2-1)$ as
region III. Proposition~\ref{sekseren}, Proposition~\ref{syveren}, and
Proposition~\ref{otteren}, respectively, takes care of region I,
region II, and region III, respectively. In Proposition~\ref{sekseren}
   we will to ease the analysis
  make the small restriction that $s_2\mid s_1$ and that $s_1
  \mid u s_2$. Furthermore, in all three propositions we assume that
  $u$ is an integer. We stress that when such assumptions do not hold
  then the formulas to be presented are still very close to be true.
What we will learn is that the
  codes ${\mbox{RM}}(S_1,S_2,u,w_1=1,w_2)$ always outperform the codes
  ${\mbox{RM}}(S,S,u^\prime,1,1)$ provided that $s_1 \geq
  4s_2$. Furthermore, for $s_1-s/s_2 \leq u$ such a result holds in
  the general situation $s_1 > s_2$.
\begin{proposition}\label{sekseren}
Consider integers  $s_1, s_2$ with $1<s_2 < s_1$. Let $u$ be an integer with $u
\leq s_1-(s_1/s_2)$. Assume $s_1/s_2$ and $us_2/s_1$ are integers and
that $s_1s_2=s^2$ for some integer $s$. Let $w_1=1$ and $w_2=s_1/s_2$
(that is, $w_1$ and $w_2$ are chosen as in
Proposition~\ref{optimal}). The code ${\mbox{RM}}(S_1,S_2,u,w_1,w_2)$
is of dimension
$$\frac{1}{2}(u^2\frac{s_2}{s_1}+u)+u\frac{s_2}{s_1}+1$$
and any ${\mbox{RM}}(S,S,u',1,1)$ of the same or larger minimum
distance is of dimension at most
$$\frac{1}{2} (\frac{s_2}{s_1}u^2+3u\sqrt{\frac{s_2}{s_1}}+2).$$
For $s_1 \geq 4 s_2$ the code ${\mbox{RM}}(S_1,S_2,u,w_1,w_2)$ is the
better one.
\end{proposition}
\begin{proof}
The dimension of ${\mbox{RM}}(S_1,S_2,u,w_1,w_2)$ is
$$\sum_{i=1}^{us_2/s_1} \frac i {s_1}{s_2}+\frac{us_2}{s_1}+1$$
and the minimum distance is $s_2(s_1-u)$. Assuming
$u^\prime=u\sqrt{\frac{s_2}{s_1}}$ is an integer, the code
${\mbox{RM}}(S,S,u^\prime,1,1)$ is of minimum distance $s_2(s_1-u)$. This
code is of dimension
$$\frac{1}{2}(\frac{s_2}{s_1}u^2+3u\sqrt{\frac{s_2}{s_1}}+2).$$
\end{proof}
\begin{proposition}\label{syveren}
Consider integers $s_1$ and $s_2$ with $s_2 < s_1$. Let $u$ be an integer
with
$$s_1-\frac{s_1}{s_2} \leq u \leq s_1-1.$$
Assume $s_1s_2=s^2$ for some integer $s$. Let $w_1=1$ and $w_2=s_1-u$
(that is, $w_1$ and $w_2$ are chosen as in
Proposition~\ref{optimal}). The dimension of
${\mbox{RM}}(S_1,S_2,u,w_1,w_2)$ equals 
\begin{equation}
s_1s_2-\frac{s_2^2(s_1-u)}{2}+s_2-\frac{s_2(s_1-u)}{2}.\label{firkant1}
\end{equation}
If $u \geq s_1-s/s_2$ then any code ${\mbox{RM}}(S,S,u',1,1)$ of the
same or larger minimum distance is of dimension at most
$$s_1s_2-\frac{(s_1-u)s_2((s_1-u)s_2-1)}{2}$$
which is less than (\ref{firkant1}) for $s_2 < s_1$. If $u < s_1
-s/s_2$ then any code  ${\mbox{RM}}(S,S,u^{\prime \prime},1,1)$ of the
same or larger minimum distance is of dimension at most
$$\frac{1}{2}\big( \frac{s_2u}{s}+2\big)\big(\frac{s_2u}{s}+1\big).$$
This number is smaller than the value
of~(\ref{firkant1}) for $s_1>4s_2$ and equal if $s_1=4s_2$.
\end{proposition}
\begin{proof}
Consider the first code which is of minimum distance $(s_1-u)s_2$. For
$i_2=s_2-1$ the value $i_1$ such that $w_1i_1+w_2i_2=u$ is
$i_1=u-(s_2-1)(s_1-u)$. Therefore the dimension equals
$$s_2(i_1+1)+\sum_{i=1}^{s_2-1}(s_1-u)i=s_1s_2-\frac{s_2^2(s_1-u)}{2}+s_2-\frac{s_2(s_1-u)}{2}.$$
If $(s_1-u)s_2 \leq s \Leftrightarrow u\geq s_1-s/s_2$ then for
$u'=2s-1-s_1s_2+us_2$ the code  ${\mbox{RM}}(S,S,u',1,1)$ is of the
same minimum distance. This code is of dimension
$$s^2-\sum_{i=1}^{(s_1-u)s_2-1}i=s_1s_2-\frac{(s_1-u)s_2((s_1-u)s_2-1)}{2}.$$
The dimension of the first code exceed the dimension of the latter
code by $(s_1-u-1)\big(s_2^2(s_1-u)/2-s_2\big)$ which is a positive
number for $u <s_1-1$ and equals zero for $u=s_1-1$.\\
If $(s_1-u)s_2 >s\Leftrightarrow u < s_1-s/s_2$ then imagining that
$s$ divides $s_2(s_1-u)$ the code  ${\mbox{RM}}(S,S,u^{\prime \prime},1,1)$ with
$$(s-u'')s=s_2(s_1-u) \Leftrightarrow u''=\frac{us_2}{s}$$
is of the same minimum distance as
${\mbox{RM}}(S_1,S_2,u,w_1,w_2)$. The dimension equals
$$\sum_{i=1}^{u''+1}i= \frac{1}{2}\big(
\frac{s_2u}{s}+2\big)\big(\frac{s_2u}{s}+1\big).$$  
Subtracting this expression from~(\ref{firkant1}) one gets a concave
function (a parabola) in $u$. Therefore the smallest value of the
difference is attained either for $u=s_1-s_1/s_2$ or for
$u=s_1-s/s_2$. Plugging in the first value and substituting $s_1=xs$,
$s_2=s/x$ one finds that the resulting function is zero for $x=2$ and
positive for $x\in ]2; s[$. Plugging in the latter value is not
needed as we already know from the first part of the theorem that here the
difference is positive.
\end{proof}
\begin{proposition}\label{otteren}
Consider integers $s_1$ and $s_2$ with $1 < s_2 < s_1$ and $s_1s_2=s^2$ where
$s$ is an integer. Let $u$ be an integer
with $s_1-1 \leq u \leq (s_1-1)+(s_2-1)$. Let $w_1=w_2=1$ (that
is, $w_1$ and $w_2$ are chosen as in
Proposition~\ref{optimal}). There is a Reed-Muller code over $S\times
S$ of the same minimum distance and the same dimension.
\end{proposition}
\begin{example}\label{extre}
This is a continuation of Example~\ref{exto}. Consider the graph in
Figure~\ref{figtvo}. 
First we take a look at the optimal weighted Reed-Muller codes corresponding to
$(s_1,s_2)=(64,16)$. For these codes region I
(Proposition~\ref{sekseren}) corresponds to rates $k/n$ below approximately
$0.5$. As $64=4 \cdot 16$ we expect the optimal weighted Reed-Muller codes
to behave very much like the corresponding Reed-Muller codes in this
region, which is indeed what the graph reveals. Considering values  
of $(s_1,s_2)$ with $s_1> 4s_2$, when 
$s_1/s_2$ increases the rates corresponding
to region I defines a smaller and smaller interval (starting of
course still with rate equal to 0). The improvements in region I
increases, but a more important contribution for the codes to become
better and better is that region II takes over at smaller
rates. Similarly, the interval of rates corresponding to region
III (Proposition~\ref{otteren}) becomes smaller and smaller. This is
the interval where the optimal
weighted Reed-Muller codes are (again) as bad as the $q$-ary Reed-Muller
codes. For $(s_1,s_2)=(64,16)$ this last mentioned interval starts at
approximately $0.875$. Already for $(s_1,s_2)=(128,8)$ the starting
point of the interval is around $0.97$.
\end{example}
Proposition~\ref{sekseren}, Proposition~\ref{syveren} and
Proposition~\ref{otteren} tell us that whenever $s_1 \geq 4s_2$ then
the optimal weighted Reed-Muller codes outperform the Reed-Muller
codes coming from $S \times S$. Example~\ref{exto} further suggests
that from that point further increasing $s_1$ and decreasing $s_2$
can only help. The following three propositions together confirm this
observation. As in previous propositions we will need to make a few assumptions on the codes
that we consider. Again we stress that when such assumptions do not hold
  then the formulas to be presented are still very close to be true.
\begin{proposition}\label{prob1}
Consider positive integers $s_1, s_2,s_1^{\prime}, s_2^{\prime}$ with
$s_1s_2=s_1^{\prime} s_2^{\prime}$, $s_1^{\prime} >s_1>s_2>s_2^{\prime}$ and
$s_1/s_2\geq 2$. Let $u$, $0<u\leq s_1-s_1/s_2$ be an integer with
$s_1 | us_2$ and consider the optimal weighted Reed-Muller code
\begin{equation}
{\mbox{RM}}(S_1,S_2,u,w_1=1,w_2=s_1/s_2).\label{eqqsnab1}
\end{equation}
If for an integer $u^{\prime}$ with $s_1^\prime | u^\prime s_2^\prime$
${\mbox{RM}}(S_1^\prime,S_2^\prime,u^\prime,w_1^\prime,w_2^\prime)$ is
    an optimal weighted Reed-Muller code of the same or smaller
    minimum distance as that of (\ref{eqqsnab1}) then the latter code
    is of dimension at least that of the first one.
\end{proposition}
\begin{proof}
If the latter code is of minimum distance close to that of
(\ref{eqqsnab1}) then it belongs to region I or II. The codes are of
minimum distance $d=s_2(s_1-u)$ and
$d^\prime=s_2^\prime(s_1^\prime-u^\prime)$, respectively. Hence,
$u^\prime \geq us_2/s_2^\prime$. If the latter code is in region I
the
improvement in dimension is at least
$$-u(\frac{1}{2}+\frac{s_2}{s_1}-\frac{s_2}{2s_2^\prime}-\frac{s_2}{s_1^\prime})$$
which is positive when $s_2+s_1/2<s_2^\prime+s_1^\prime/2$. Writing
$s_1^\prime=\mu s_1$ this corresponds to
\begin{equation}
\mu^2(s_1/2)+\mu(-s_2-s_1/2)+s_2>0. \label{eqqstar}
\end{equation}
The left side is a convex parabola with roots $\mu=1$ and $\mu=2s_2/s_1$. The assumption $s_1/s_2\geq 2$ therefore
guarantees that~(\ref{eqqstar}) holds for all $\mu >1$.\\
Assume next that the latter code in the proposition is in region II. The
improvement can be calculated to be
\begin{equation}
u^2(-\frac{s_2}{2s_1})+u(\frac{s_2s_2^{\prime}}{2}+\frac{s_2}{2}-\frac{s_2}{2s_1})+(\frac{s_1s_2}{2}-\frac{s_2^{\prime}
  s_1s_2}{2}+s_2^{\prime} -1) \label{eqqdiamond}
\end{equation}
which is a concave function in $u$. Our assumptions give
$s_1-s_1/s_2^\prime \leq u \leq s_1-s_1/s_2$ and therefore it is
enough to plug $u=s_1-s_1/s_2^\prime$ and $u= s_1-s_1/s$ into
(\ref{eqqdiamond}) and then to check that the resulting values are
positive. The first value is positive if 
\begin{equation}
\mu^3(-\frac{s_1}{2s_2})+\mu^2(1+\frac{s_1}{2s_2}+\frac{s_1}{2})+\mu(-1-\frac{s_1}{2}-s_2)+s_2 \label{eqqsnabel}
\end{equation}
is positive. The roots of this function in $\mu$ are $0$, $2s_2/s_1$,
and $s_2$. Hence, (\ref{eqqsnabel}) is indeed positive for $\mu \in
]1,s_2[$. When $u=s_1-s_1/s_2$ is plugged into~(\ref{eqqdiamond}) we
get $$\frac{s_2}{\mu}+\frac{s_1s_2}{2}-s_2-\frac{s_1s_2}{2\mu}$$
which is positive for $\mu >1$.
\end{proof}
\begin{proposition}
Consider positive integers $s_1,s_2,s_1^\prime,s_2^\prime$ with
$s_1s_2=s_1^\prime s_2^\prime$, $s_1^\prime >s_1>s_2>s_2^\prime$, and
$s_1/s_2\geq 2$. Let $u$, $s_1-s_1/s_2 \leq u \leq s_1-1$ be an
integer and consider the optimal weighted Reed-Muller code
\begin{equation}
{\mbox{RM}}(S_1,S_2,u,w_1=1,w_2=s_1-u). \label{eqqtrekant}
\end{equation}
If for an integer $u^\prime$ ${\mbox{RM}}(S_1^\prime, S_2^\prime
,u^\prime, w_1^\prime,w_2^\prime)$ is an optimal weighted Reed-Muller
code of the same or smaller minimum distance as that
of~(\ref{eqqtrekant}) then the latter code is of dimension at least
that of the first one. 
\end{proposition}
\begin{proof}
If the latter code is of minimum distance close to that
of~(\ref{eqqtrekant}) then it belongs to region II. As in the proof of
the preceding proposition we have $u^\prime \geq u
s_2/s_2^\prime$. The improvement in dimension can be calculated to be
at least $$(s_2-s_2^\prime)(\frac{s_1s_2}{2}-\frac{s_2u}{2}-1)$$ which
takes on its minimal value $s_2/2-1$ for $u=s_1-1$. Combining this
with the assumption $s_2>s_2^\prime \geq 1$ proves the proposition. 
\end{proof}
\begin{proposition}
Consider positive integers $s_1, s_2, s_1^\prime , s_2^\prime$ with
$s_1s_2=s_1^\prime s_2^\prime $, $s_1^\prime >s_1
>s_2>s_2^\prime$. Let $u$, $s_1-1 \leq u \leq (s_1-1)+(s_2-1)$ be an
integer and consider the optimal weighted Reed-Muller code
\begin{equation}
{\mbox{RM}}(S_1, S_2, u, w_1=1, w_2=1). \label{eqqtilde} 
\end{equation}
If for an integer $u^\prime$
${\mbox{RM}}(S_1^\prime,S_2^\prime,u^\prime,w_1^\prime=1, w_2^\prime)$
is an optimal weighted Reed-Muller code of the same or smaller minimum
distance as that of~(\ref{eqqtilde}) then the latter code is of
dimension at least that of the first one.
\end{proposition}
\begin{proof}
The latter code either belongs to region II or III. For those in 
region III the result is pretty obvious so we consider only codes in
region II. Let $d$ be the minimum distance of the code
in~(\ref{eqqtilde}). We have $u^\prime
\geq (s_1s_2-d)/s_2^\prime$. The improvement in dimension can be
calculated to be at least
$$\frac{1}{2}d^2+d(-1-\frac{s_2^\prime}{2})+s_2^\prime.$$
We may assume $d >s_2^\prime$ as we are in region II and the result follows.
\end{proof}
The construction of weighted Reed-Muller codes is very concrete, but
for completeness we should mention that it is not the most optimal. Consider instead the codes $E({\mathbb{M}},{\mathcal{S}})$ with ${\mathcal{S}}=S_1 \times
\cdots \times S_m$ and 
\begin{equation}
{\mathbb{M}}=\{ X_1^{i_1} \cdots X_m^{i_m} \mid (s_1-i_1)\cdots
(s_m-i_m) \geq \delta\}. \label{defmjc}
\end{equation}
Among the codes with designed distance $\delta$
(Theorem~\ref{afstand}) these are the codes of highest possible
dimension. When $S_1=\cdots =S_m={\mathbf{F}}_q$ holds the
construction simply is that of Massey-Costello-Justesen codes
 (see~\cite{masseycostellojustesen} and \cite{kabatianski}). 
\section{Dual codes}
As is well-known, for the special case of ${\mathcal{S}}={\mathbf{F}}_q
\times \cdots \times {\mathbf{F}}_q$ the duals of $q$-ary Reed-Muller
codes, weighted Reed-Muller codes, and Massey-Costello-Justesen codes,
respectively, are $q$-ary Reed-Muller codes, weighted Reed-Muller
codes, and hyperbolic codes, respectively \cite{abs},
\cite{hyperbolic} (for the definition of
hyperbolic codes we refer to (\ref{eqhyp}) below). 
More examples of codes
$E({\mathbb{M}},{\mathcal{S}})$ where similar neat correspondences hold can be
found in~\cite{brasosul}. Turning to a general point ensemble
${\mathcal{S}}=S_1 \times \cdots \times S_m$, however, it does not in
general hold that the dual of a weighted Reed-Muller code
is again a weighted Reed-Muller code. 
Nor does it hold in general that the dual of a Massey-Costello-Justesen
code is a hyperbolic code. For a simple counter example which fits the
description of a weighted Reed-Muller code as well as the description
of a Massey-Costello-Justesen code consider the ordinary Reed-Solomon
code over ${\mathcal{S}}=\mathbf{F}_q^\ast$ and recall that
${\mbox{ev}}_{\mathcal{S}}(1)$ is not a parity check for this
particular code.\\
Fortunately, for the class of codes
\begin{eqnarray}
E^\perp({\mathbb{M}},\mathcal{S})&=\{ \vec{c} \in
{\mathbf{F}}_q^{n=|\mathcal{S}|}\mid \vec{c} \cdot
{\mbox{ev}}_{\mathcal{S}}(M)=0, {\mbox{ for all }} M \in {\mathbb{M}} \} \nonumber
\end{eqnarray}
we have a technique similar to that of Section~\ref{secto} to estimate
the minimum distance. This technique is known as the Feng-Rao bound
(\cite{FR1}, \cite{FR2}). We now
recall this bound following the description of Shibuya
and Sakaniwa in~\cite{shibuya}. Consider the
following definition of a linear code.
\begin{definition}
Let $B=\{ {\vec{b}}_1, \ldots ,  {\vec{b}}_n\}$ be a basis
for ${\mathbf{F}}_{q}^{n}$ and let $G\subseteq B$. We define  $C(B,G)={\mbox{Span}}_{\mathbf{F}_{q}}\{\vec{b} \mid
\vec{b} \in G\}$. The dual code  is denoted
$C^{\perp}(B,G)$.
\end{definition}
The Feng-Rao bound calls for the following set of spaces.
\begin{definition}
Let $L_{-1}=\emptyset$, $L_0=\{ \vec{0} \}$ and
$L_l={\mbox{Span}}_{\mathbf{F}_{q}}\{ \vec{b}_1, \ldots ,\vec{b}_{l}\}$ for $l=1, \ldots ,n$.
\end{definition}
We obviously have a chain of spaces $\{\vec{0} \}=L_0 \subsetneq L_1
\subsetneq \cdots \subsetneq L_{n-1} \subsetneq L_n={\mathbf{F}}_q^n$. Hence,
we can define a function as follows.
\begin{definition}\label{defwwb}
Define $\bar{\rho}: {\mathbf{F}}_{q}^{n} \rightarrow \{0, 1, \ldots ,n\}$ by
$\bar{\rho}(\vec{v})=l$ if $\vec{v} \in L_l \backslash L_{l-1}$. 
\end{definition}
\begin{definition}
Let $I=\{1,
2, \ldots , n\}$. An ordered pair $(i,j) \in I^2$ is said to be well-behaving
if $\bar{\rho}(\vec{b}_u \ast \vec{b}_v) < \bar{\rho}(\vec{b}_i \ast \vec{b}_j)$
for all $u$ and $v$ with $1 \leq u \leq i, 1 \leq v \leq j$ and $(u,v)
\neq (i, j)$. Here, $\ast$ is the componentwise product. 
\end{definition}
\begin{definition}
For $l=1, \ldots , n$ define
$$\bar{\mu}(l)=\# \{ (i,j) \in I^2 \mid (i,j){\mbox{ is well-behaving and }}
\bar{\rho}(\vec{b}_i \ast \vec{b}_j)=l\}.$$  
\end{definition}
The Feng-Rao bound now is (\cite[Prop.\ 1]{shibuya}):
\begin{theorem}
The minimum distance of $C^{\perp}(B,G)$ is at least
$$\min \{ \bar{\mu} (l) \mid \vec{b}_l \notin G\}.$$
\end{theorem}
Turning to the codes $E(\mathbb{M},\mathcal{S})$ we enumerate the basis 
$$ \{{\mbox{ev}}_{\mathcal{S}}(X_1^{i_1} \cdots X_m^{i_m}) \mid 0 \leq
i_1 < s_1, \ldots , 0\leq i_m < s_m\}$$
according to a total degree lexicographic ordering on the  monomials
$X_1^{i_1}\cdots X_m^{i_m}$. For
${\vec{b}}_l={\mbox{ev}}_{\mathcal{S}}(X_1^{i_1}\cdots X_m^{i_m})$,
 $0 \leq i_1 < s_1, \ldots ,0 \leq i_m < s_m$, we clearly have
$\bar{\mu}(l) \geq (i_1+1)\cdots(i_m+1)$. Hence, by the Feng-Rao bound the
minimum distance of $E^\perp({\mathbb{M}},\mathcal{S})$ is at least
\begin{eqnarray}
\min \{(i_1+1)\cdots (i_m+1) \mid X_1^{i_1}\cdots X_m^{i_m} \notin
{\mathbb{M}}, 0 \leq i_1 < s_1, \ldots , 0\leq i_m < s_m\}. \label{eqfrao}
\end{eqnarray}
Consider the code
$$({\mbox{RM}}(S_1, \ldots , S_m, (s_1-1) \cdots
(s_m-1)-u-\epsilon,w_1,\ldots , w_m))^\perp$$
where $\epsilon$ is a very small positive number. The
bound~(\ref{eqfrao}) tells us that the minimum distance is at least
that of the weighted Reed-Muller code ${\mbox{RM}}(S_1, \ldots ,
S_m,$ $u,w_1, \ldots ,w_m)$. Observe that the codes are of the same
dimension. Similarly, the hyperbolic code 
$E^\perp({\mathbb{M}},\mathcal{S})$ 
\begin{eqnarray}
\mathbb{M}&=&\{X_1^{i_1}\cdots X_m^{i_m} \mid (i_1+1)\cdots (i_m+1) <
\delta\} \label{eqhyp}
\end{eqnarray}
has designed minimum distance equal to $\delta$ just as the
Massey-Costello-Justesen code in~(\ref{defmjc}). Again, the two codes
are of the same dimension. \\
The Feng-Rao bound comes with a decoding algorithm that corrects up to
half the designed minimum distance~\cite{FR2,handbook}. This algorithm of
course applies in particular to the above dual codes. \\
The remaining part of the paper is concerned with
decoding algorithms for the codes $E({\mathbb{M}},\mathcal{S})$
including the codes from Section~\ref{secwrm}.
\section{Subfield subcode decoding}\label{secsubdec}
As already noted by Kasami et al. in~\cite{kasami}, any ordinary $q$-ary
Reed-Muller code (in the terminology of the present paper this means a
$q$-ary Reed-Muller code from ${\mathcal{S}}={\mathbf{F}}_q \times
\cdots \times {\mathbf{F}}_q$) can be seen as a subfield subcode of a Reed-Solomon
code. The Reed-Solomon code will be over the field
${\mathbf{F}}_{q^m}$ and is constructed by evaluating polynomials of
degree at most $uq^{m-1}$ in the $q^m$ different elements of
${\mathbf{F}}_{q^m}$. The above observation guarantees that codes
$E({\mathbb{M}},\mathcal{S})$ in general can be seen as subcodes of subfield
subcodes of certain Reed-Solomon codes over ${\mathbf{F}}_{q^m}$, but
it is not straightforward which elements of ${\mathbf{F}}_{q^m}$ to
use. This problem, however, is easy to overcome if we use the approach
by Santhi~\cite{santhi}.\\ 
Let  $\{\vec{b}_1, \ldots , \vec{b}_n\}$ be a basis for
${\mathbb{F}}_{q^m}$ as a vectorspace over ${\mathbf{F}}_q$. Following
Santhi we now define a map $\varphi : {\mathbb{F}}_q^m\rightarrow
{\mathbb{F}}_{q^m}$ by
\begin{equation}
\varphi(a_1, \ldots , a_m) =a_1\vec{b}_1+\cdots +a_m\vec{b}_m \label{eqsnab1}
\end{equation}
and note that $\varphi (a_1, \ldots ,
a_m)^{q^v}=a_1\vec{b}_1^{q^v}+\cdots +a_m\vec{b}_m^{q^v}$. Writing
$X=\varphi(a_1, \ldots , a_m)$ we have
$$
\left[ \begin{array}{ccc}
\vec{b}_1&\cdots & \vec{b}_m\\
\vec{b}_1^q&\cdots & \vec{b}_m^q\\
\vdots&\ddots&\vdots\\
\vec{b}_1^{q^{m-1}}&\cdots & \vec{b}_m^{q^{m-1}}
\end{array}\right]
\left[ \begin{array}{c}
a_1\\
a_2\\
\vdots\\
a_m
\end{array}
\right]
=
\left[ \begin{array}{c}
X\\
X^q\\
\vdots \\
X^{q^{m-1}}
\end{array}
\right].
$$
By~\cite[Cor.\ 2.38]{niederreiter} the matrix on the left side is
invertible and therefore there exist polynomials $F_1, \ldots , F_m
\in {\mathbb{F}}_{q^m}[T]$ such that $a_i=F_i(X)$. The polynomials
$F_1, \ldots , F_m$ do not depend on the point $(a_1, \ldots , a_m)$
under consideration. From this observation one deduces that
$F(P)=F(F_1(\varphi(P)), \ldots  ,F_m(\varphi(P)))$ for all $P \in
{\mathbb{F}}_q^m$.
\begin{theorem}\label{tobefilledin}
Write ${\mathcal{S}}=\{P_1, \ldots , P_n\}$. The code $E({\mathbb{M}},{\mathcal{S}})$ is a
subcode of a subfield subcode of the Reed-Solomon code over
${\mathbf{F}}_{q^m}$ which is constructed by evaluating polynomials of
degree at most
\begin{equation}
t=\max \{\deg M \mid M \in {\mathbb{M}} \} \label{eqttt}
\end{equation}
in the elements $\varphi(P_1),\ldots , \varphi(P_n)$. Here $\varphi$
is the function in (\ref{eqsnab1}).
\end{theorem}
Following the Pellikaan-Wu approach~\cite{pw_ieee} we can now decode
$E({\mathbb{M}},{\mathcal{S}})$ by applying the Guruswami-Sudan list decoding
algorithm to the corresponding Reed-Solomon code \cite{GS} and by performing a
few additional steps. The complexity of the Guruswami-Sudan list
decoding algorithm is in the literature often claimed to be
${\mathcal{O}}(n^3)$.  For more precise statements of the decoding complexity which
takes the multiplicity into account we refer
to~\cite{beelenbrander}. The Guruswami-Sudan algorithm corrects up to
$n(1-\sqrt{R})$ errors of the Reed-Solomon code. It is therefore clear that the above approach can
decode up to
\begin{equation}
\lceil n (1-\sqrt{\frac{t q^{m-1}+1}{n}})\rceil \label{eqtrknt}
\end{equation} 
errors of $E({\mathbb{M}},{\mathcal{S}})$ (Here, $t$ is as in (\ref{eqttt})). This is indeed a fine result for many
codes $E({\mathbb{M}},{\mathcal{S}})$. However, it is also clear that for other
choices of $E({\mathbb{M}},{\mathcal{S}})$ (\ref{eqtrknt})
may be close to zero or even negative. For the particular case of an
optimal weighted Reed-Muller code ${\mbox{RM}}(S_1,S_2,u,w_1=1,w_2)$
(\ref{eqtrknt}) becomes 
\begin{equation}
\lceil s_1s_2\big(1-\sqrt{\frac{uq+1}{s_1s_2}}\big)\rceil.\label{rad1}
\end{equation}
(recall from Proposition~\ref{optimal} that  $w_1 \leq w_2$ always holds). 
If $s_1s_2$ is close to $q^2$ and $u$ is not too large then this bound
guarantees list decoding. If $s_2$ is much smaller than $s_1$ then the
bound may not even guarantee that the algorithm can correct  a single
error. This is the reason why we in the present paper consider also a
second decoding algorithm. Before getting to the second algorithm, we
apply the first one to the Joyner code. 
\subsection{The Joyner code}\label{secjoyner}
Toric codes were introduced by Hansen in \cite{hansen} and further
generalized by Joyner in \cite{joyner}, by Ruano in~\cite{ruano1,ruano2}
and by Little et al.\  in~\cite{little}. Among the most famous toric codes is
the $[49,11,28]$ code over ${\mathbf{F}}_8$ presented
in~\cite[Ex.\ 3.9]{joyner}. This code is known as the Joyner
code. Attempts have been made to decode it, but without much luck
so far. We now demonstrate how to decode it even beyond its minimum distance by applying the method of
this section in combination with a small trick. \\
The Joyner code originally was introduced in the language of
polytopes. Alternatively, one can define it \cite{ruano2} as a code
$E({\mathbb{M}},S)$
where
$${\mathcal{S}}={\mathbf{F}}_8^{\ast} \times {\mathbf{F}}_8^{\ast}=\{(x_1,y_1),
\ldots ,(x_{49},y_{49})\},$$
$${\mathbb{M}}=\{1\} \cup \{X^iY^j \mid 1 \leq i,j {\mbox{ \ and \ }}
i+j \leq 5\}.$$
Let the polynomial corresponding to a given code word $\vec{c}$ be
$$F(X,Y)=F_{0,0}+\sum_{\begin{array}{c}i,j \geq 1\\
i+j\leq 5  \end{array}}F_{i,j}X^iY^j.$$
Let $\vec{r}=\vec{c}+\vec{e}$ be the received word. Assume for a
moment that we know $F_{0,0}$. We then subtract $(F_{0,0}, \ldots
,F_{0,0})$ from $\vec{r}$ to get a word $\vec{r^\prime}$ that in the
error free positions corresponds to
$$\sum_{\begin{array}{c}i,j \geq 1\\
i+j\leq 5  \end{array}}F_{i,j}X^iY^j.$$
For $i=1, \ldots , 49$ we now divide the $i$th entry of $\vec{r^\prime}$ with $x_iy_i$ to
produce a word $\vec{r^{\prime \prime}}$. Observe, that this is doable
because $x_i,y_i \neq 0$ holds. The word $\vec{r^{\prime \prime}}$ in
the error free positions corresponds to
$$\sum_{\begin{array}{c}i,j \geq 1\\
i+j\leq 5  \end{array}}F_{i,j}X^{i-1}Y^{j-1}.$$
We have $\vec{r^{\prime \prime}}=\vec{c^{\prime \prime}}+\vec{e^{\prime
      \prime}}$ where $\vec{c^{\prime \prime}} \in
  E({\mathbb{M^{\prime \prime}}},{\mathcal{S}})$,
  ${\mathbb{M^{\prime \prime}}}=\{X^iY^j \mid i+j\leq 3\}$ and $\vec{e^{\prime
      \prime}}$ is non-zero in exactly the same positions as
    $\vec{e}$. The Reed-Muller code $E(\mathbb{M^{\prime \prime}},{\mathcal{S}})$
   is a subfield subcode of a $[49,25,25]$ Reed-Solomon code over
   ${\mathbf{F}}_{64}$. The exact form of the Reed-Solomon code is
   described by Theorem~\ref{tobefilledin}. Given a code word from the
   output of the Reed-Solomon list decoder we multiply for $i=1,
   \ldots , 49$ the $i$th entry with $x_iy_i$ and add $F_{0,0}$.\\
Of course we do not as assumed above know $F_{0,0}$ in
advance. Therefore we must try out all $8$ possible values of this
number. The error correction capability of the corresponding algorithm
is described in Table~\ref{tabny1}. We see that we can correct up to
$31$ errors even though the minimum distance of the Joyner code is
only $28$.  
   \begin{table}
    \caption{Error correction capability when the multiplicity used of
      the Reed-Solomon code 
      decoder is $m$}
    \label{tabny1}
    \begin{center}
     \begin{tabular}{rcccccc}
      \toprule
$m$&1&2&3&4&5&6\\
    \cmidrule(l{.25em}r{.25em}){2-7}
capability& 12&20&24&27&29&31\\
\bottomrule
    \end{tabular}
    \end{center}
   \end{table}
\section{An interpretation of the Guruswami-Sudan list decoding algorithm}
The second decoding algorithm of the present paper is a direct
interpretation of the Guruswami-Sudan list decoding algorithm. We
build on works by Pellikaan et al.\ \cite{pw_ieee}, and Augot et al.\
\cite{augot}, \cite{augotstepanov} who concentrated on $q$-ary
Reed-Muller codes and Reed-Solomon product codes. We consider general
code $E({\mathbb{M}},{\mathcal{S}})$ and improve on the above mentioned work by establishing
new information on how many zeros of prescribed multiplicity a
polynomial can have when given information about its leading monomial
with respect to the lexicographic ordering. In combination with a
preparation step this will allow us to correct more errors. The idea of a
preparation step comes from~\cite{geilmatsumoto}. The improved
information regarding the zeros is derived by strengthening results
reported by Dvir et al. in~\cite{dvir}. This is done in
Subsection~\ref{secbounding}. In Subsection~\ref{secalg} we present
the algorithm and elaborate on its decoding radius. 
\subsection{Bounding the number of zeros of multiplicity
  $r$}\label{secbounding}
The definition of multiplicity that we will use relies on the Hasse
derivative. Before recalling the definition of the
Hasse derivative let us fix some notation. Assume we have a vector of
variables $\vec{X}=(X_1, \ldots ,X_m)$ and a vector $\vec{k}=(k_1,
\ldots , k_m)\in {\mathbf{N}}_0^m$ then we will write
$\vec{X}^{\vec{k}}=X_1^{k_1} \cdots X_m^{k_m}$. 
In the following
${\mathbf{F}}$ is any field. 
\begin{definition}
Given $F(\vec{X})\in {\mathbf{F}}[\vec{X}]$ and
$\vec{k} \in {\mathbf{N}}_0^m$ the $\vec{k}$'th
Hasse derivative of $F$, denoted by $F^{(\vec{k})}(\vec{X})$ is the
coefficient of $\vec{Z}^{\vec{k}}$ in $F(\vec{X}+\vec{Z})$. In other words 
$$F(\vec{X}+\vec{Z})=\sum_{\vec{k}} F^{(\vec{k})}(\vec{X})\vec{Z}^{\vec{k}}.$$
\end{definition}
The concept of multiplicity for univariate polynomials is generalized
to multivariate polynomials in the following way.
\begin{definition}\label{defmult}
For $F(\vec{X}) \in {\mathbf{F}}[\vec{X}]\backslash \{ {0} \}$ and
$\vec{a}\in {\mathbf{F}}^m$ we define the multiplicity of $F$ at $\vec{a}$
denoted by ${\mbox{mult}}(F,\vec{a})$ as follows: Let $M$ be an
integer such that for every $\vec{k}=(k_1, \ldots ,
k_m) \in {\mathbf{N}}_0^m$ with $k_1+\cdots +k_m < M$, $F^{(\vec{k})}(\vec{a})=0$
  holds, but for some  $\vec{k}=(k_1, \ldots ,
k_m) \in {\mathbf{N}}_0^m$ with $k_1+\cdots +k_m = M$,
$F^{(\vec{k})}(\vec{a})\neq 0$ holds, then 
${\mbox{mult}}(F,\vec{a})=M$. If $F=0$ then we define ${\mbox{mult}}(F,\vec{a})=\infty$.
\end{definition}
The Schwartz-Zippel bound with multiplicity was reported already 
in~\cite{augot}, \cite{augotstepanov} but was only recently proved,
~\cite{dvir}. It goes as follows:
\begin{theorem}\label{SZ-lemma}
Let $F(\vec{X}) \in {\mathbf{F}}[\vec{X}]$ be a non-zero polynomial of
total degree $u$. Then for any finite set $S \subseteq {\mathbf{F}}$
\begin{eqnarray}
\sum_{\vec{a}\in S^m}{\mbox{mult}}(F,\vec{a}) \leq u | S |^{{m-1}}. \nonumber
\end{eqnarray}
\end{theorem}
We have the following useful corollary:
\begin{corollary}\label{corsz}
Let $F(\vec{X}) \in {\mathbf{F}}[\vec{X}]$ be a
non-zero polynomial of total degree $u$ and let $S \subseteq
{\mathbf{F}}$ be finite.
The number
of zeros of $F$ of multiplicity at least $r$ from $S^m$ is at most
\begin{equation}
\frac{u}{r}|S |^{m-1}.\label{eqcorsz}
\end{equation}
\end{corollary}
For the $q$-ary
Reed-Muller codes 
$${\mbox{RM}}_q(u,m)={\mbox{RM}}({\mathbf{F}}_q, \ldots ,
{\mathbf{F}}_q,u,1,\ldots ,1),$$
 Pellikaan and Wu in~\cite{pw_ieee} presented two decoding algorithms, 
a subfield subcode decoding algorithm and a direct interpretation of
the Guruswami-Sudan algorithm. The analysis of the latter relies on
\cite[Lem.\ 2.4, Lem.\ 2.5]{PellikaanWu} which combines to the
  following result:
\begin{proposition}\label{pwbound}
Consider a polynomial $F(\vec{X}) \in {\mathbf{F}}[\vec{X}]$ of total
degree $u$, $u < rq$ and define $w=\lfloor u/q\rfloor$. The number of points in ${\mathbf{F}}_q^m$ where $F$ has at least
multiplicity $r$ is at most equal to
\begin{eqnarray}
\frac{{m+r-1 \choose m}q^m+(u-qw){m+r-w-2 \choose
    m-1}q^{m-1}-{m+r-w-1 \choose m}q^m}{{m+r-1 \choose r-1 }}.\label{eqboundpw}
\end{eqnarray}
\end{proposition}
Augot and Stepanov~\cite{augot} gave an improved estimate on the
decoding radius of the latter algorithm (the direct interpretation of
the Gurswami-Sudan algorithm) by using instead
Corollary~\ref{corsz}. We here present a direct proof that indeed,
Corollary~\ref{corsz} is stronger than Proposition~\ref{pwbound}. 
\begin{proposition}
For all $u \in  [ 0,rq-1]$ it hold that (\ref{eqcorsz}) is
smaller than or equal to (\ref{eqboundpw}).
\end{proposition}
\begin{proof}
We consider the two expressions as
functions in $u$ on the interval $[0,rq]$. Our first observation is
that (\ref{eqboundpw}) is a continuously piecewise linear function,
each piece corresponding to a particular value of $w$. The
corresponding $r$ slopes constitute a decreasing sequence. Combining this
observation with the fact that (\ref{eqcorsz})  is linear in $u$ and with
the fact that the two expressions are the same at each of the end
points of the interval proves the result.
\end{proof}
As a preparation step to improve upon Theorem~\ref{SZ-lemma} and Corollary~\ref{corsz} we start
by generalizing them. 
We will need a couple of results
from~\cite[Sec.\ 2]{dvir}. The first corresponds to~\cite[Lem.\
5]{dvir}.
\begin{lemma}\label{lem5}
Consider $F(\vec{X}) \in {\mathbf{F}}[\vec{X}]$ and
$\vec{a} \in {\mathbf{F}}^m$. For any $\vec{k}=(k_1, \ldots ,k_m)
\in {\mathbf{N}}_0^m$ we have 
$${\mbox{mult}}(F^{(\vec{k})},\vec{a}) \geq mult(F,\vec{a})-(k_1+ \cdots +k_m).$$
\end{lemma}
The next result that we recall corresponds to the last part of
\cite[Proposition 6]{dvir}. 
\begin{proposition}\label{propsammensat}
 Given $F(X_1, \ldots , X_m) \in {\mathbf{F}}[X_1, \ldots
 ,X_m]$ and
$$Q(Y_1, \ldots , Y_l)=(Q_1(\vec{Y}), \ldots  ,Q_m(\vec{Y}))
\in {\mathbf{F}}[Y_1, \ldots , Y_l]^m$$ let $F \circ Q$ be the
polynomial $F(Q_1(\vec{Y}), \ldots , Q_m(\vec{Y}))$.  For any
$\vec{a}\in {\mathbf{F}}^l$ we have $${\mbox{mult}}(F \circ Q,\vec{a})
\geq {\mbox{mult}}(F,Q(\vec{a})).$$
\end{proposition}
We get the following Corollary, which is closely related
to~\cite[Corollary 7]{dvir}. 
\begin{corollary}\label{cor8}
Let $F(X_1, \ldots , X_m) \in {\mathbf{F}}[X_1, \ldots , X_m]$ and
$
\vec{b}_1, \ldots  ,\vec{b}_{m-1},\vec{c} \in {\mathbf{F}}^m$ be
given. Write $F^\ast(T_1, \ldots  ,T_{m-1}) =F(T_1\vec{b}_1 + \cdots +
T_{m-1}\vec{b}_{m-1}+\vec{c})$.
For any $(t_1, \ldots , t_{m-1}) \in {\mathbf{F}}^{m-1}$ we
have
\begin{multline*}
{\mbox{mult}}(F^\ast(T_1, \ldots , T_{m-1}),(t_1, \ldots , t_{m-1})) \\
\geq{\mbox{mult}}(F(X_1, \ldots ,X_m),
t_1\vec{b}_1+\cdots +t_{m-1}\vec{b}_{m-1}+\vec{c}).
\end{multline*}
\end{corollary}
Let $\prec$ be the lexicographic ordering on the set of monomials in
variables $X_1, \ldots , X_m$ such that $X_m\prec \cdots \prec X_1$
holds. We now write
$$F(X_1, \ldots , X_m)=\sum_{j_1, \ldots  ,j_{m-1}} 
X_1^{j_1}\cdots X_{m-1}^{j_{m-1}}F_{j_1, \ldots
  j_{m-1}} (X_m).$$
Let $X_1^{i_1} \cdots X_m^{i_m}$ be the leading monomial of $F$ with
respect to $\prec$. Then
due to the definition of $\prec$, $F_{i_1, \ldots , i_{m-1}}(X_m)$ is a (univariate) polynomial
of degree $i_m$.
For $a_m \in {\mathbf{F}}$  define  
$$
r(a_m)={\mbox{mult}}(F_{i_1, \ldots , i_{m-1}}(X_m),a_m).$$
Clearly, 
\begin{equation}
\sum_{a_m \in S_m}r(a_m) \leq i_m. \label{eqfoerste}
\end{equation}
We have 
$$
F^{(0,\ldots , 0,r(a_m))}(X_1, \ldots , X_m)=\sum_{j_1, \ldots ,
  j_{m-1}}
X_1^{j_1} \cdots X_{m-1}^{j_{m-1}}F_{j_1, \ldots
  ,j_{m-1}}^{(r(a_m))}(X_m)$$
and due to the definition of $\prec$ and to the definition of
$r(a_m)$ we have
\begin{equation}
{\mbox{lm}}_{\prec}(F^{(0,\ldots , 0,r(a_m))}(X_1, \ldots ,
X_{m-1},a_m))=X_1^{i_1}\cdots X_{m-1}^{i_{m-1}}. \label{eqnaeste}
\end{equation}
Applying first Lemma~\ref{lem5} with $\vec{k}=(0,\ldots ,0,r(a_m))$ and 
afterwards   
Corollary~\ref{cor8} with 
$\vec{b}_1=(1,0,\ldots ,0), \ldots ,
\vec{b}_{m-1}=(0, \ldots , 0,1,0)$, $\vec{c}=(0, \ldots , 0,a_m)$ and $t_1=a_1, \ldots ,
t_{m-1}=a_{m-1}$ we get the following result which is
closely related to a result in~\cite[Proof of Lemma 8]{dvir}:
\begin{align}
&{\mbox{mult}}\big(F(X_1, \ldots  ,X_m),(a_1, \ldots , a_m)\big)\nonumber \\
&\leq (0+ \cdots +0+r(a_m))+{\mbox{mult}}\big(F^{(0,\ldots
  ,0,r(a_m))}(X_1, \ldots ,X_m),(a_1, \ldots , a_m)\big)\nonumber\\
&\leq r(a_m) +{\mbox{mult}}\big(F^{(0,\ldots  ,0,r(a_m))}(X_1, \ldots
,X_{m-1},a_m),(a_1, \ldots  ,a_{m-1})\big). \label{eqtredie}
\end{align} 
We are now ready to generalize Theorem~\ref{SZ-lemma}. Let in the
remaining part of this subsection $S_1, \ldots , S_m$ be finite
subsets of arbitrary field ${\mathbf{F}}$. Also we will relax from the
assumption that $s_1 \geq \cdots \geq s_m$.
\begin{theorem}\label{prop-sz-gen}
Let $F(\vec{X}) \in {\mathbf{F}}[\vec{X}]$ be a non-zero polynomial and
let ${\mbox{lm}}(F)=X_1^{i_1} \cdots X_m^{i_m}$ be its leading
monomial with respect to a lexicographic ordering. Then for
any finite sets $S_1, \ldots ,S_m \subseteq {\mathbf{F}}$
\begin{eqnarray}
\sum_{\vec{a}\in S_1 \times \cdots \times S_m}{\mbox{mult}}(F,\vec{a})
\leq i_1s_2\cdots s_m+s_1i_2s_3 \cdots s_m+\cdots +s_1\cdots s_{m-1}i_m.\nonumber
\end{eqnarray}
\end{theorem}
\begin{proof}
We prove the theorem for the monomial
ordering $\prec$. Dealing with general 
lexicographic orderings is simply a question of relabeling the
variables. Clearly the theorem holds for
$m=1$. For $m>1$ we consider~(\ref{eqtredie}). Assuming the
theorem holds when the number of variables is smaller than $m$ we
get by applying~(\ref{eqfoerste}) and~(\ref{eqnaeste}) the following estimate
\begin{align}
&\sum_{\vec{a} \in S_1 \times \cdots \times
  S_m}{\mbox{mult}}(F,\vec{a}) \nonumber \\
&\leq i_ms_1 \cdots s_{m-1}+s_m(i_1s_2
\cdots s_{m-1}+\cdots +i_{m-1}s_1\cdots s_{m-2})\nonumber \\
&=i_1s_2\cdots s_m+i_2s_1 s_3\cdots s_m+\cdots i_ms_1\cdots
s_{m-1}\nonumber
\end{align}
as required.
\end{proof}
We have the following immediate generalization of Corollary~\ref{corsz}.
\begin{corollary}\label{cor-sz-gen}
Let $F(\vec{X}) \in {\mathbf{F}}[\vec{X}]$ be a non-zero polynomial and
let ${\mbox{lm}}(F)=X_1^{i_1} \cdots X_m^{i_m}$ be its leading
monomial with respect to a lexicographic ordering. Assume $S_1, \ldots
,S_m \subseteq {\mathbf{F}}$ are finite sets. 
Then over
$S_1 \times \cdots \times S_m $ the number
of zeros of multiplicity at least $r$ is less than or equal to 
\begin{eqnarray}
\big( i_1s_2\cdots s_m+s_1i_2s_3\cdots s_m+\cdots +s_1\cdots
s_{m-1}i_m \big) /r.\label{szbound}
\end{eqnarray}
\end{corollary}
The analysis leading to Theorem~\ref{SZ-lemma} suggests the following
function to more accurately estimate the number of zeros of multiplicity at most $r$
of a polynomial with leading monomial $X_1^{i_1}\cdots X_m^{i_m}$:
\begin{definition}\label{defD}
Let $r \in {\mathbf{N}}, i_1, \ldots , i_m \in {\mathbf{N}}_0$. Define 
$$D(i_1,r,s_1)=\min \big\{\big\lfloor \frac{i_1}{r} \big\rfloor,s_1\big\}$$
and for $m \geq 2$
\begin{multline*}
D(i_1, \ldots , i_m,r,s_1, \ldots ,s_m)=
\\
\begin{split}
\max_{(u_1, \ldots  ,u_r)\in A(i_m,r,s_m) }&\bigg\{ (s_m-u_1-\cdots -u_r)D(i_1,\ldots ,i_{m-1},r,s_1,
\ldots ,s_{m-1})\\
&\quad+u_1D(i_1, \ldots , i_{m-1},r-1,s_1, \ldots ,s_{m-1})+\cdots
\\
&\quad +u_{r-1}D(i_1, \ldots ,i_{m-1},1,s_1, \ldots , s_{m-1})+u_rs_1\cdots
s_{m-1} \bigg\}
\end{split}
\end{multline*}
where 
\begin{multline}
A(i_m,r,s_m)= \nonumber \\
\{ (u_1, \ldots , u_r) \in {\mathbf{N}}_0^r \mid u_1+ \cdots
+u_r \leq s_m {\mbox{ \ and \ }} u_1+2u_2+\cdots +ru_r \leq i_m\}.\nonumber
\end{multline}
\end{definition}
\begin{theorem}\label{prorec}
For a polynomial $F(\vec{X})\in {\mathbf{F}}[\vec{X}]$ let $X_1^{i_1}\cdots X_m^{i_m}$ be its leading monomial with
respect to $\prec$ (this is the lexicographic ordering with $X_m\prec \cdots \prec X_1$). Then $F$ has at most $D(i_1, \ldots , i_m,r,s_1,
\ldots ,s_m)$ zeros of multiplicity at least $r$ in $S_1\times \cdots
\times S_m$. The corresponding recursive algorithm produces a number
that is at most equal to the number found in
Corollary~\ref{cor-sz-gen} and is at most equal to $s_1 \cdots s_m$.
\end{theorem}
\begin{proof}
The proof of the first part of the proposition is an induction proof. The result clearly holds for
$m=1$. Given $m>1$ assume it holds for $m-1$. 
For $d=1, \ldots ,
r-1$ let $u_d$ be the number of $a_m$'s with $r(a_m)=d$ and let $u_r$
be the number of $a_m$'s with $r(a_m) \geq r$. The number of $a_m$'s
with $r(a_m)=0$ is $s_m-u_1-\cdots -u_r$. The boundary conditions that
$u_1+\cdots +u_r \leq s_m$ and $u_1+2u_2+\cdots +ru_r \leq i_m$ are
obvious. For every $a_m$ with $r(a_m)=d$, $d=0, \ldots , r-1$ for
$(a_1, \ldots , a_m)$ to be a zero of multiplicity at least $r$ the
last expression in~(\ref{eqtredie}) must be at least $r-d$. For
$a_m$ with $r(a_m)\geq r$ all choices of $a_1, \ldots , a_{m-1}$ are
legal. This proves the first part of the proposition. As both
Corollary~\ref{cor-sz-gen} and the above proof rely on~(\ref{eqtredie}), Theorem~\ref{prorec} cannot produce a
number greater than what is found in Corollary~\ref{cor-sz-gen}. The
condition $u_1+\cdots +u_r \leq s_m$ and the definition of
$D(i_1,r,s_1)$ imply the last result.
\end{proof}
It only makes sense to apply the function $D(i_1, \ldots ,i_m,r,s_1,
\ldots , s_m)$ to monomials $ X_1^{i_1}\cdots X_m^{i_m}$ in 
$$\Delta(r,s_1,\ldots ,s_m)=\{ X_1^{i_1}\cdots X_m^{i_m} \mid  \lfloor
i_1/s_1\rfloor+\cdots +\lfloor i_m/s_m\rfloor < r\}.
$$
\begin{proposition}\label{rembig}
Assume $X_1^{i_1}\cdots X_m^{i_m} \notin \Delta(r,s_1, \ldots ,
s_m)$. Then there exists a
polynomial with leading monomial  $X_1^{i_1}\cdots X_m^{i_m}$ such that all elements of $S_1\times \cdots \times S_m$ are zeros of
multiplicity at least $r$.
\end{proposition}
\begin{example}
In a
number of experiments listed in~\cite{hmpage} we calculated the value $D(i_1, \ldots ,
i_m,r,q, \ldots ,q)$ for various choices of $m$, $q$ and $r$ and for all values of
$(i_1, \ldots , i_m)$ such that $X_1^{i_m} \cdots X_m^{i_m} \in
\Delta(r,q, \ldots , q)$.
Here we list the mean improvement in comparison
with the situation where 
Corollary~\ref{corsz} is applied. More formally, we list in
Table~\ref{tabendnuenny1} for various fixed $q, r, m$ the mean value of
\begin{equation}
\frac{ \min \{(i_1+\cdots + i_m)q^{m-1}/r,q^m\}-D(i_1, \ldots ,
i_m,r,q, \ldots , q)}{\min \{(i_1+\cdots +
i_m)q^{m-1}/r,q^m\}}.\label{eqangle}
\end{equation}
\begin{table}[!h]
\centering
\caption{The mean value of (\ref{eqangle}); truncated.}
\newcommand{\SP}{~~}
\begin{tabular}{@{}c@{}c@{~~~~}r@{.}l@{\SP}r@{.}l@{\SP}r@{.}l@{\SP}r@{.}l@{\SP}r@{.}l@{\SP}r@{.}l@{\SP}r@{.}l@{\SP}r@{.}l@{\SP}r@{.}l@{\SP}r@{.}l@{}}
\toprule
$m$&& \multicolumn{8}{c}{2} & \multicolumn{8}{c}{3} & \multicolumn{4}{c}{4} \\
\cmidrule(r){3-10} \cmidrule(r){11-18} \cmidrule{19-22}
$r$&&\multicolumn{2}{l}{2}&\multicolumn{2}{l}{3}&\multicolumn{2}{l}{4}&\multicolumn{2}{l}{5}&\multicolumn{2}{l}{2}&\multicolumn{2}{l}{3}&\multicolumn{2}{l}{4}&\multicolumn{2}{l}{5}&\multicolumn{2}{l}{2}&\multicolumn{2}{l}{3}\\
\addlinespace
\multirow{7}{*}{$q$}
&\multicolumn{1}{c}{2}&0&363&0&273&0&337&0&291&0&301&0&300&0&342&0&307&0&248&0&260\\
&\multicolumn{1}{c}{3}&0&217&0&286&0&228&0&236&0&194&0&224&0&213&0&214&0&158&0&177\\
&\multicolumn{1}{c}{4}&0&191&0&197&0&232&0&195&0&158&0&169&0&180&0&172&0&125&0&135\\
&\multicolumn{1}{c}{5}&0&155&0&167&0&174&0&197&0&139&0&145&0&148&0&153&0&110&0&116\\
&\multicolumn{1}{c}{7}&0&128&0&137&0&138&0&138&0&119&0&122&0&121&0&119&0&093&0&098\\
&\multicolumn{1}{c}{8}&0&126&0&127&0&134&0&126&0&114&0&115&0&113&0&111&0&089&0&093\\
\bottomrule
\end{tabular}
\label{tabendnuenny1}
\end{table}
Despite the significant mean improvement, according to our experiments in~\cite{hmpage} for most fixed degrees $u$ there are examples of
exponents $(i_1, \ldots , i_m)$, $i_1+\cdots +i_m=u$ such that $D(i_1,
\ldots ,i_m,r,s_1, \ldots , s_m)=\lfloor (i_1+\cdots +i_m)q^{m-1}/r\rfloor$.
\end{example}
Sometimes the values $D(i_1, \ldots ,i_m,r, s_1, \ldots , s_m)$ may be
time consuming to
calculate. Therefore it is relevant to have some closed formula
estimates of these numbers. We next present such estimates for the
case of two variables. Note, that the following proposition covers all
monomials in $\Delta(r,s_1,s_2)$.
\begin{proposition}\label{protwovar}
For $k=1, \ldots , r-1$,  $D(i_1,i_2,r,s_1,s_2)$ is upper bounded by\\
$\begin{array}{cl}
{\mbox{(C.1)}}&  {\displaystyle{s_2\frac{i_1}{r}+\frac{i_2}{r}\frac{i_1}{r-k}}}\\
&{\mbox{if \  }}(r-k)\frac{r}{r+1}s_1 \leq i_1 < (r-k)s_1
{\mbox{ \ and \ }} 0\leq i_2 <ks_2,\\
{\mbox{(C.2)}}&
  {\displaystyle{s_2\frac{i_1}{r}+((k+1)s_2-i_2)(\frac{i_1}{r-k}-\frac{i_1}{r})+(i_2-ks_2)(s_1-\frac{i_1}{r})}}\\
& {\mbox{if \ }}(r-k)\frac{r}{r+1}s_1 \leq i_1 < (r-k)s_1 {\mbox{ \
    and \ }} ks_2\leq i_2 <(k+1)s_2,\\
{\mbox{(C.3)}}&
{\displaystyle{s_2\frac{i_1}{r}+\frac{i_2}{k+1}(s_1-\frac{i_1}{r})}}\\
&{\mbox{if \ }} (r-k-1)s_1 \leq i_1 < (r-k)\frac{r}{r+1}s_1 {\mbox{ \
    and \ }} 0 \leq i_2 < (k+1)s_2.
\end{array}
$\\
Finally,\\
$\begin{array}{cl}
{\mbox{(C.4)}}& {\displaystyle{D(i_1,i_2,r,s_1,s_2)=s_2\lfloor \frac{i_1}{r} \rfloor
  +i_2(s_1-\lfloor \frac{i_1}{r} \rfloor )}}\\
& {\mbox{if \ }} s_1(r-1) \leq i_1 < s_1r {\mbox{ \ and \ }} 0 \leq i_2 < s_2.
\end{array}
$\\
The above numbers are at most equal to $\min\{(i_1s_2+s_1i_2)/r, s_1s_2 \}$.
\end{proposition}
\begin{proof}First we consider the values of $i_1, i_2,
r,s_1,s_2$ corresponding to one of the cases (C.1), (C.2), (C.3). Let
$k$ be the largest number (as in Proposition~\ref{protwovar})
  such that $i_1 < (r-k)s_1$. Indeed $k \in \{1, \ldots ,
r-1\}$. We have
\begin{multline}
\label{eqsnabelen}
D(i_1,i_2,r,s_1,s_2) \leq\\
\max_{(u_1, \ldots  ,u_r)\in B(i_2,r,s_2)} \bigg\{
s_2\frac{i_1}{r}+u_1(\frac{i_1}{r-1}-\frac{i_1}{r})+\cdots
+u_k(\frac{i_1}{r-k}-\frac{i_1}{r})\\
+u_{k+1}(s_1-\frac{i_1}{r})+ \cdots +u_r(s_1-\frac{i_1}{r})\bigg\}
\end{multline}
where 
\begin{multline*}
B(i_2,r,s_2)=\{(u_1, \ldots , u_r) \in {\mathbf{Q}}^r \mid 0 \leq u_1,
\ldots ,u_r,\\
u_1+\cdots +u_r \leq s_2, u_1+2u_2+
\cdots +ru_r \leq i_2\}.
\end{multline*}
We observe, that 
$$k(\frac{i_1}{r-l}-\frac{i_1}{r})\leq l
(\frac{i_1}{r-k}-\frac{i_1}{r})$$ 
holds for $l \leq k$. Furthermore, we have the biimplication
\begin{eqnarray}
(r-k)\frac{r}{r+1} s_1 
\leq i_1
\Leftrightarrow (k+1)(\frac{i_1}{r-k}-\frac{i_1}{r}) \geq
k(s_1-\frac{i_1}{r}). \nonumber
\end{eqnarray}
Therefore, if the conditions in (C.1) are satisfied then~(\ref{eqsnabelen}) takes on its maximum when
$u_k=\frac{i_2}{k}$ and the remaining $u_i$'s equal $0$. If the
conditions in (C.2) are satisfied then (\ref{eqsnabelen}) takes on its
maximum at $u_k=(k+1)s_2-i_2$, $u_{k+1}=(i_2-ks_2)$ and the remaining
$u_i$'s equal $0$. If the conditions in (C.3) are satisfied
then~(\ref{eqsnabelen}) takes on its maximal value at
$u_{k+1}=\frac{i_2}{k+1}$ and the remaining $u_i$'s equal $0$.\\
Finally, if $s_1(r-1) \leq i_1 < s_1r$ and $0 \leq i_2 \leq s_2$ then
$D(i_1,i_2,r,s_1,s_2)$ is the maximal value of
$$s_2\lfloor \frac{i_1}{r} \rfloor +u_1(s_1-\lfloor \frac{i_1}{r}
\rfloor)+ \cdots +u_r(s_1-\lfloor \frac{i_1}{r}
\rfloor)$$
over $B(i_2,r,s_2)$. The maximum is attained for $u_1=i_2$ and all other
$u_i$'s equal $0$. The proof of the last result follows the proof of
the last part of Theorem~\ref{prorec}.
\end{proof}
\begin{remark}
Experiments show (see~\cite{hmpage}) that 
 the numbers produced by Proposition~\ref{protwovar} are often much
smaller than $\min \{ (i_1s_2 +s_1i_2)/r,s_1s_2\}$. However, there are
cases where they are identical. This happens for example when
$i_1=s_1(r-1)$ and $r$ divides $s_1$ and $s_2$. In the
proof of (C.1), (C.2), (C.3) we allowed $u_1, \ldots , u_r$ to be
rational numbers rather than integers. Therefore we cannot expect the upper bounds in
Proposition~\ref{protwovar} to equal the true value of
$D(i_1,i_2,r,s_1,s_2)$ in general. Our experiments show that the
bounds in (C.1), (C.2), (C.3) are sometimes close to
$D(i_1,i_2,r,s_1,s_2)$ but not always. Hence the best information is
found by actually applying the function $D(i_1,i_2,r,s_1,s_2)$ directly.
\end{remark}
\subsection{The decoding algorithm}\label{secalg}
\noindent The main ingredient of the decoding algorithm is to find an
interpolation polynomial 
$$Q(X_1, \ldots , X_m,Z)=Q_0(X_1, \ldots , X_m)+Q_1(X_1,\ldots ,
X_m)Z+\cdots +Q_t(X_1, \ldots ,X_m)Z^t$$
such that $Q(X_1, \ldots , X_m,F(X_1, \ldots  ,X_m))$ cannot have more
than $n-E$ different zeros of multiplicity at least $r$ whenever
${\mbox{Supp}}(F)\subseteq {\mathbb{M}}$. The integer $E$ above is the
number of errors to be corrected by our list decoding algorithm. In
\cite{pw_ieee}, \cite{augot}, \cite{augotstepanov} this requirement is
described in terms of bounds on the total degree of the polynomials $Q_i$. As
we will use improved information that depends not on total degree but on the leading monomial with respect to a lexicographic
ordering the situation becomes more complicated. To
fulfill the above requirement we will define appropriate sets of monomials
$B(i,E,r)$, $i=1, \ldots , t$ and then require $Q_i(X_1, \ldots
,X_m)$ to be chosen such that ${\mbox{Supp}}(Q_i)\subseteq
B(i,E,r)$. Rather than using the results from the previous section on
all possible choices of $F(X_1, \ldots  ,X_m)$ with ${\mbox{Supp}}(F)
\subseteq {\mathbb{M}}$ we need only consider the worst cases where the
leading monomial of $F$ is contained in the following set:
\begin{definition}
$$\overline{{\mathbb{M}}}=\{ M \in {\mathbb{M}} \mid {\mbox{ \ if \ }} N\in
{\mathbb{M}} {\mbox{ \ and \ }} M|N {\mbox{ \ then \ }} M=N\}.$$
\end{definition}
Hence, $\overline{{\mathbb{M}}}$ is so to speak the border of
${\mathbb{M}}$. 
\begin{definition}
Given positive integers $i,E,r$ with $E<n$ let
$$B(i,E,r)=\{K \in \Delta(r,s_1, \ldots ,s_m ) \mid D_r(K M^{i}) < n-E {\mbox{ \ for
    all \ }} M \in \overline{{\mathbb{M}}} \}.$$
Here $D_r(X_1^{i_1}, \ldots , X_m^{i_m})$ can either be $D(i_1, \ldots
, i_m,r,s_1, \ldots ,s_m)$ or in the case of two variables it can be
the numbers from Proposition~\ref{protwovar}. Another option would be
to let $D_r(X_1^{i_1}, \ldots , X_m^{i_m})$ be the number
in~(\ref{szbound}). 
\end{definition}
The decoding algorithm calls for positive integers $t,E,r$ such that
\begin{equation}
\sum_{i=1}^t |B(i,E,r)| >n N(m,r), \label{eqsnabel}
\end{equation}
where $N(m,r)={{m+r}\choose{m+1}}$ is the number of linear equations to be satisfied
for a point in ${\mathbf{F}}_q^{m+1}$ to be a zero of $Q(X_1, \ldots ,
X_m,Z)$ of multiplicity at least $r$. 
As we will see condition~(\ref{eqsnabel}) ensures that we
can correct $E$ errors. We say that $(t,E,r)$ satisfies the
initial condition if given the pair $(E,r)$, $t$ is the smallest integer such
that~(\ref{eqsnabel}) is satisfied. Whenever this is the case we
define $B^\prime(t,E,r)$ to be any subset of $B(t,E,r)$ such that 
$$\sum_{i=1}^{t-1}|B(i,E,r)| + |B^\prime(t,E,r)|=n N(m,r)+1.$$
Replacing $B(t,E,r)$ with $B^\prime(t,E,r)$ will lower the run-time
of the algorithm.
\begin{Algorithm}\label{thealgorithm}
\noindent {\it{Input:}}\\
Received word $\vec{r}=(r_1, \ldots , r_n) \in {\mathbf{F}}_q^n$.\\
Set of integers $(t,E,r)$ that satisfies the initial condition.\\
Corresponding sets $B(1,E,r)\, \ldots ,B(t-1,E,r),B^\prime(t,E,r)$.\\

\noindent {\it{Step 1}}\\
Find non-zero polynomial
$$Q(X_1, \ldots , X_m Z)=Q_0(X_1, \ldots , X_m)+Q_1(X_1, \ldots
,X_m)Z+\cdots +Q_t(X_1, \ldots ,X_m)Z^t$$
such that
\begin{enumerate}
\item ${\mbox{Supp}}(Q_i) \subseteq B(i,E,r)$ for $i=1, \ldots ,t-1$
  and ${\mbox{Supp}}(Q_t) \subseteq B^\prime(t,E,r)$,
\item $(P_i,r_i)$ is a zero of $Q(X_1, \ldots , X_m,Z)$ of
  multiplicity at least $r$ for $i=1, \ldots ,n$.
\end{enumerate}

\noindent {\it{Step 2}}\\
Find all $F(X_1, \ldots ,X_m) \in {\mathbf{F}}_q[X_1, \ldots , X_m]$
such that 
\begin{equation}
(Z-F(X_1, \ldots ,X_m)) |Q(X_1, \ldots , X_m,Z).\label{eqstar}
\end{equation}
\noindent {\it{Output:}}\\
A list containing $(F(P_1), \ldots , F(P_n))$ for all $F$ satisfying~(\ref{eqstar}).
\end{Algorithm}

\begin{theorem}
The output of Algorithm~\ref{thealgorithm} contains all words in
$E({\mathbb{M}},{\mathcal{S}})$ within distance $E$ from the received word
$\vec{r}$. Once the preparation step has been performed the algorithm runs in time ${\mathcal{O}}(\bar{n}^3)$ where
$\bar{n}=n{{m+r}\choose{m+1}}$. For given multiplicity $r$ the maximal
number of correctable errors $E$ and the corresponding sets $B(1,
E,r), \ldots ,B(t-1,E,r)$, $B^\prime(t,E,r)$ can be found in time
${\mathcal{O}}(n \log(n) r^m s' |\overline{\mathbb{M}}|/\sigma)$ 
assuming that the values of the function $D_r$ are known. Here
$\sigma = \max\{\deg M \mid M \in
\overline{\mathbb{M}}\}$ and $s' = \max\{s_1, \dotsc, s_m\}$.
\end{theorem}
\begin{proof} The interpolation problem corresponds to $\bar{n}$
  homogeneous linear equations in $\bar{n}+1$ unknowns. 
Hence, indeed a suitable $Q$ can be found in time
${\mathcal{O}}(\bar{n}^3)$. Now assume ${\mbox{Supp}}(F) \subseteq {\mathbb{M}}$  and that
${\mbox{dist}}_H({\mbox{ev}}_{\mathcal{S}}(F),\vec{r})\leq E$. Then $P_j$ is a zero of $Q(X_1, \ldots , X_m,F(X_1,
\ldots , X_m))$ of multiplicity at least $r$ for at least $n-E$ choices of $j$. By the
definition of $B(i,E,r)$ this can, however, only be the case if $Q(X_1, \ldots ,
X_m,F(X_1, \ldots , X_m))=0$. Therefore, $Z-F(X_1, \ldots ,X_m)$ is a
factor in $Q(X_1, \ldots ,X_m,Z)$. Finding linear factors of
polynomials in $({\mathbf{F}}_q[X_1, \ldots , X_m])[Z]$ can be done in
time ${\mathcal{O}}(\bar{n}^3)$ by applying Wu's algorithm
in~\cite{wu} (see \cite[p.\ 20]{PellikaanWu}).\end{proof}
Algorithm~\ref{thealgorithm} works for general codes
$E({\mathbb{M}},\mathcal{S})$ and for any of the three possible
choices of $D_r(X_1^{i_1} \cdots X_m^{i_m})$ as described prior to the
algorithm. In such a general setting it is impossible to say anything
reasonable regarding the decoding radius. The algorithm apparently
works best for not too large code dimensions. With this in mind we
restrict the analysis to optimal weighted Reed-Muller codes
${\mbox{RM}}(S_1, S_2,u,w_1,w_2)$ in region I. That is, we assume
$w_1=1$, $w_2=s_1/s_2$ and $u \leq s_1-s_1/s_2$. As the function
$D(i_1,i_2,r,s_1,s_2)$ is highly irregular and Proposition~\ref{protwovar}
contains four  quite different cases it seems impossible to perform the
analysis for other choices than $D(i_1,i_2)=(i_1s_2+i_2s_1)/r$ which
corresponds to the weakest version of the decoding algorithm.
\begin{proposition}\label{proradius}
Consider an optimal weighted Reed-Muller code
${\mbox{RM}}(S_1,S_2,u,w_1=1,w_2=s_1/s_2)$ with $s_2 |s_1$ and $u\leq s_1-s_1/s_2$ a
positive integer. When equipped with $D(i_1,i_2)=(i_1s_2+i_2s_1)/r$
the decoding radius of Algorithm~\ref{thealgorithm} is at least 
\begin{equation}
s_1s_2(1-\sqrt[3]{u/s_1}).\label{rad2}
\end{equation}
\end{proposition}
\begin{proof}
Let $v$ be divisible by $u$. The number of variables in the
interpolation polynomial when $t=\deg_Z Q$ is chosen to be $v/u$ is
lower bounded by
\begin{eqnarray}
&&\sum_{j=0}^{v/u-1}\big[
(v+1-ju)+\frac{1}{2}(v+1-ju)((v-ju)s_2/s_1-1)\big]\nonumber \\
&>&\frac{1}{2} \frac{s_2}{s_1}\sum_{i=1}^{v/u}(ui)^2 \geq
\frac{s_2}{s_1}\frac{v^3}{6u}. \nonumber
\end{eqnarray}
The number of equations is $s_1s_2r(r+1)(r+2)/6$ and therefore
$$v \geq \sqrt[3]{ur(r+1)(r+2)s_1^2}
$$
is a sufficient condition for the existence of an interpolation
polynomial. Assume
\begin{eqnarray}
&&E < s_1s_2\big(1-\sqrt[3]{u(1+1/r)(1+2/r)/s_1}\big)\nonumber \\
&\Downarrow \nonumber \\
&&E < s_1s_2-\frac{1}{r} s_2\sqrt[3]{ur(r+1)(r+2)s_1^2}.\nonumber 
\end{eqnarray}
Substituting $v=\sqrt[3]{ur(r+1)(r+2)s_1^2}$ we get $r(s_1s_2-E)
>vs_2$ which ensures that $Q(X_1,X_2(F(X_1,X_2))=0$ for any codeword
$\vec{c}={\mbox{ev}}_{\mathcal{S}}(F)$ within distance $E$ from
$\vec{r}$. Letting $r$ go to infinity finishes the proof.
\end{proof}
Comparing the decoding radii~(\ref{rad1}) and (\ref{rad2}) we conclude that
when $s_2$ is close to $q$ then the subfield subcode decoder is
superior. On the other hand when $s_2$ is much smaller than $q$ then
the decoding algorithm of the present section performs best. 
\begin{example}
In this example we investigate the performance of
Algorithm~\ref{thealgorithm} when applied to optimal weighted
Reed-Muller codes and Massey-Costello-Justesen codes coming from the
point ensembles 
${\mathcal{S}}=S_1\times S_2$ with $s_1=64$, $s_2=8$, and $s_1=256$,
$s_2=16$, respectively. Our findings are presented in
Table~\ref{tabcode1} and Table~\ref{tabcode2}, respectively. The
decoding capability is calculated for different choices of
$D_r(X_1^{i_1}X_2^{i_2})$ and different multiplicities $r$. The symbol $S$, $C$,
and $D$, respectively, corresponds to $D_r(X_1^{i_1}X_2^{i_2})$ being chosen as the
Schwartz-Zippel bound~(\ref{szbound}), the closed formulas of
Proposition~\ref{protwovar}, and the function $D(i_2,i_1,r,s_2,s_1)$,
respectively. The letter $W$ stands for optimal weighted Reed-Muller
code and $I$ means the Massey-Costello-Justesen code of the same
minimum distance. Further $u$ is the third argument in the notion
${\mbox{RM}}(S_1,S_2,u,w_1=1,w_2=s_1/s_2)$ and $d$ is the minimum
distance. $Sub$ stands for the estimated decoding radius~(\ref{eqtrknt}) of the
algorithm in Section~\ref{secsubdec} and $Dim$ is the dimension of the
code. For large values of $r$ the calculations regarding
$D(i_1,i_2,r,s_1,s_2)$ become quite heavy and have therefore not been
made. We can see from the tables that for the considered codes
Algorithm~\ref{thealgorithm} outperforms the subfield subcode approach
from Section~\ref{secsubdec}. In some cases it decodes much more than half the
minimum distance. It is apparent that the function
$D(i_1,i_2,r,s_1,s_2)$ as well as the closed formula
expressions of Proposition~\ref{protwovar} help bringing up the error correction
capability in comparison with the situation where the
Schwartz-Zippel bound~(\ref{szbound}) is used. It is clear that the small
gain in dimension by considering Massey-Costello-Justesen
codes rather than optimal weighted Reed-Muller codes comes with a
heavy price as Algorithm~\ref{thealgorithm} corrects much fewer
errors. By inspection the estimation of decoding radius from Proposition~\ref{proradius} seems to be quite
close to what is found by our computer experiments.
\begin{table}
\centering
    \newcommand{\cmidrulesA}{
        \cmidrule(lr){1-1}
        \cmidrule(lr){2-2}
        \cmidrule(lr){3-4}
        \cmidrule(lr){5-6}
        \cmidrule(lr){7-8}
        \cmidrule(lr){9-10}
        \cmidrule(lr){11-12}
        \cmidrule(lr){13-14}
    }
    \caption{Table of error correction capability for optimal weighted Reed-Muller codes and
    Massey-Costello-Justesen codes when $s_1=64$ and
    $s_2=8$.}
\label{tabcode1}
    \begin{tabular}{r r c c c c c c c c c c c c}
    \toprule
    \multicolumn{2}{r}{$u$/$d$}
    & 3 & 488 & 4 & 480 & 7 & 456 & 15 & 392 & 16 & 384 & 20 & 352 \\
    \cmidrule(lr){3-4}
    \cmidrule(lr){5-6}
    \cmidrule(lr){7-8}
    \cmidrule(lr){9-10}
    \cmidrule(lr){11-12}
    \cmidrule(lr){13-14}
    $r$ & Bound & W & I & W & I & W & I & W & I & W & I & W & I
    \\
    \cmidrulesA
    \multirow{3}{*}{2}
    & S
    & \multicolumn{2}{c}{267}
    & \multicolumn{2}{c}{243}
    & \multicolumn{2}{c}{191}
    & 103 & \phantom{0}95 & \phantom{0}95 & \phantom{0}87 & \phantom{0}67 & \phantom{0}59
    \\
    & C
    & \multicolumn{2}{c}{286}
    & \multicolumn{2}{c}{266}
    & \multicolumn{2}{c}{219}
    & 131 & 128 & 122 & 119 & \phantom{0}97 & \phantom{0}94
    \\
    & D
    & \multicolumn{2}{c}{298}
    & \multicolumn{2}{c}{277}
    & \multicolumn{2}{c}{228}
    & 135 & 131 & 121 & 119 & \phantom{0}99 & \phantom{0}95
    \\
    \cmidrulesA
    \multirow{3}{*}{3}
    & S
    & \multicolumn{2}{c}{287}
    & \multicolumn{2}{c}{263}
    & \multicolumn{2}{c}{213}
    & 130 & 122 & 122 & 117 & \phantom{0}95 & \phantom{0}90
    \\
    & C
    & \multicolumn{2}{c}{301}
    & \multicolumn{2}{c}{279}
    & \multicolumn{2}{c}{234}
    & 149 & 145 & 138 & 135 & 113 & 109
    \\
    & D
    & \multicolumn{2}{c}{319}
    & \multicolumn{2}{c}{298}
    & \multicolumn{2}{c}{255}
    & 177 & 175 & 161 & 160 & 139 & 135
    \\
    \cmidrulesA
    \multirow{3}{*}{4}
    & S
    & \multicolumn{2}{c}{295}
    & \multicolumn{2}{c}{273}
    & \multicolumn{2}{c}{225}
    & 145 & 139 & 139 & 131 & 111 & 105
    \\
    & C
    & \multicolumn{2}{c}{307}
    & \multicolumn{2}{c}{286}
    & \multicolumn{2}{c}{242}
    & 159 & 155 & 147 & 145 & 123 & 118
    \\
    & D
    & \multicolumn{2}{c}{328}
    & \multicolumn{2}{c}{311}
    & \multicolumn{2}{c}{269}
    & 196 & 195 & 181 & 181 & 160 & 159
    \\
    \cmidrulesA
    \multirow{2}{*}{9}
    & S
    & \multicolumn{2}{c}{312}
    & \multicolumn{2}{c}{292}
    & \multicolumn{2}{c}{247}
    & 173 & 166 & 166 & 159 & 140 & 134
    \\
    & C
    & \multicolumn{2}{c}{318}
    & \multicolumn{2}{c}{299}
    & \multicolumn{2}{c}{255}
    & 178 & 173 & 169 & 166 & 144 & 139
    \\
    \cmidrulesA
    \multirow{2}{*}{20}
    & S
    & \multicolumn{2}{c}{320}
    & \multicolumn{2}{c}{301}
    & \multicolumn{2}{c}{258}
    & 185 & 178 & 178 & 171 & 153 & 147
    \\
    & C
    & \multicolumn{2}{c}{323}
    & \multicolumn{2}{c}{304}
    & \multicolumn{2}{c}{262}
    & 188 & 182 & 180 & 175 & 155 & 149
    \\
    \midrule
    & Sub
    & \multicolumn{2}{c}{198}
    & \multicolumn{2}{c}{149}
    & \multicolumn{2}{c}{33}
    & \multicolumn{2}{c}{0}
    & \multicolumn{2}{c}{0}
    & \multicolumn{2}{c}{0}
    \\
    \cmidrule(lr){2-2}
    \cmidrule(lr){3-4}
    \cmidrule(lr){5-6}
    \cmidrule(lr){7-8}
    \cmidrule(lr){9-10}
    \cmidrule(lr){11-12}
    \cmidrule(lr){13-14}
    & $\floor{\frac{d-1}{2}}$
    & \multicolumn{2}{c}{243}
    & \multicolumn{2}{c}{239}
    & \multicolumn{2}{c}{227}
    & \multicolumn{2}{c}{195}
    & \multicolumn{2}{c}{191}
    & \multicolumn{2}{c}{175}
    \\
    \cmidrule(lr){2-2}
    \cmidrule(lr){3-4}
    \cmidrule(lr){5-6}
    \cmidrule(lr){7-8}
    \cmidrule(lr){9-10}
    \cmidrule(lr){11-12}
    \cmidrule(lr){13-14}
    & Dim
    & \multicolumn{2}{c}{4}
    & \multicolumn{2}{c}{5}
    & \multicolumn{2}{c}{8}
    & 24 & 25 & 27 & 28 & 39 & 41
    \\
    \bottomrule
    \end{tabular}
\end{table}

\newcommand{\cmidrulesB}{
    \cmidrule(lr){1-1}
    \cmidrule(lr){2-2}
    \cmidrule(l{0.5em}r{0.5em}){3-4}
    \cmidrule(l{0.0em}r{0.5em}){5-6}
    \cmidrule(l{0.0em}r{0.5em}){7-8}
    \cmidrule(l{0.0em}r{0.5em}){9-10}
    \cmidrule(l{0.0em}r{0.5em}){11-12}
    \cmidrule(l{0.0em}r{0.5em}){13-14}
}
\begin{table}
    \centering
    \caption{Table of error correction capability for optimal weighted Reed-Muller codes and
    Massey-Costello-Justesen codes when $s_1=256$ and
    $s_2=16$.}
    \label{tabcode2}
    \newcommand{\SP}{~~}
    \begin{tabular}{r r c@{\SP}c@{\SP}c@{\SP}c@{\SP}c@{\SP}c@{\SP}c@{\SP}c@{\SP}c@{\SP}c@{\SP}c@{\SP}c}
    \toprule
    \multicolumn{2}{r}{$u$/$d$}
    & 5 & 4016 & 8 & 3968 & 15 & 3856 & 31 & 3600 & 36 & 3620 & 55 & 3216 \\
    \cmidrule(l{0.5em}r{0.5em}){3-4}
    \cmidrule(l{0.0em}r{0.5em}){5-6}
    \cmidrule(l{0.0em}r{0.5em}){7-8}
    \cmidrule(l{0.0em}r{0.5em}){9-10}
    \cmidrule(l{0.0em}r{0.5em}){11-12}
    \cmidrule(l{0.0em}r{0.5em}){13-14}
    $r$ & Bound & W & I & W & I & W & I & W & I & W & I & W & I
    \\
    \cmidrulesB
    \multirow{3}{*}{2}
    & S
    & \multicolumn{2}{c}{2591}
    & \multicolumn{2}{c}{2335}
    & \multicolumn{2}{c}{1927}
    & 1359 & 1335 & 1231 & 1207 & \phantom{0}839 & \phantom{0}791
    \\
    & C
    & \multicolumn{2}{c}{2680}
    & \multicolumn{2}{c}{2456}
    & \multicolumn{2}{c}{2112}
    & 1565 & 1557 & 1392 & 1391 & 1022 & 1003
    \\
    & D
    & \multicolumn{2}{c}{2729}
    & \multicolumn{2}{c}{2504}
    & \multicolumn{2}{c}{2153}
    & 1589 & 1583 & 1411 & 1408 & 1035 & 1015
    \\
    \cmidrulesB
    \multirow{3}{*}{3}
    & S
    & \multicolumn{2}{c}{2714}
    & \multicolumn{2}{c}{2479}
    & \multicolumn{2}{c}{2106}
    & 1578 & 1551 & 1455 & 1434 & 1082 & 1034
    \\
    & C
    & \multicolumn{2}{c}{2790}
    & \multicolumn{2}{c}{2579}
    & \multicolumn{2}{c}{2240}
    & 1695 & 1684 & 1552 & 1547 & 1190 & 1167
    \\
    & D
    & \multicolumn{2}{c}{2861}
    & \multicolumn{2}{c}{2651}
    & \multicolumn{2}{c}{2326}
    & 1859 & 1855 & 1707 & 1706 & 1359 & 1351
    \\
    \cmidrulesB
    \multirow{2}{*}{4}
    & S
    & \multicolumn{2}{c}{2779}
    & \multicolumn{2}{c}{2555}
    & \multicolumn{2}{c}{2195}
    & 1691 & 1667 & 1575 & 1551 & 1211 & 1163
    \\
    & C
    & \multicolumn{2}{c}{2843}
    & \multicolumn{2}{c}{2635}
    & \multicolumn{2}{c}{2305}
    & 1782 & 1767 & 1638 & 1632 & 1284 & 1260
    \\
    \cmidrulesB
    \multirow{2}{*}{9}
    & S
    & \multicolumn{2}{c}{2894}
    & \multicolumn{2}{c}{2689}
    & \multicolumn{2}{c}{2362}
    & 1895 & 1871 & 1784 & 1763 & 1443 & 1367
    \\
    & C
    & \multicolumn{2}{c}{2928}
    & \multicolumn{2}{c}{2730}
    & \multicolumn{2}{c}{2415}
    & 1935 & 1919 & 1811 & 1804 & 1469 & 1442
    \\
    \cmidrulesB
    \multirow{2}{*}{20}
    & S
    & \multicolumn{2}{c}{2947}
    & \multicolumn{2}{c}{2751}
    & \multicolumn{2}{c}{2439}
    & 1988 & 1966 & 1882 & 1862 & 1551 & 1506
    \\
    & C
    & \multicolumn{2}{c}{2964}
    & \multicolumn{2}{c}{2772}
    & \multicolumn{2}{c}{2464}
    & 2007 & 1989 & 1894 & 1884 & 1562 & 1529
    \\
    \midrule
    & Sub
    & \multicolumn{2}{c}{1806}
    & \multicolumn{2}{c}{1199}
    & \multicolumn{2}{c}{130}
    & \multicolumn{2}{c}{0}
    & \multicolumn{2}{c}{0}
    & \multicolumn{2}{c}{0}
    \\
    \cmidrule(lr){2-2}
    \cmidrule(l{0.5em}r{0.5em}){3-4}
    \cmidrule(l{0.0em}r{0.5em}){5-6}
    \cmidrule(l{0.0em}r{0.5em}){7-8}
    \cmidrule(l{0.0em}r{0.5em}){9-10}
    \cmidrule(l{0.0em}r{0.5em}){11-12}
    \cmidrule(l{0.0em}r{0.5em}){13-14}
    & $\floor{\frac{d-1}{2}}$
    & \multicolumn{2}{c}{2007}
    & \multicolumn{2}{c}{1983}
    & \multicolumn{2}{c}{1927}
    & \multicolumn{2}{c}{1799}
    & \multicolumn{2}{c}{1759}
    & \multicolumn{2}{c}{1607}
    \\
    \cmidrule(lr){2-2}
    \cmidrule(l{0.5em}r{0.5em}){3-4}
    \cmidrule(l{0.0em}r{0.5em}){5-6}
    \cmidrule(l{0.0em}r{0.5em}){7-8}
    \cmidrule(l{0.0em}r{0.5em}){9-10}
    \cmidrule(l{0.0em}r{0.5em}){11-12}
    \cmidrule(l{0.0em}r{0.5em}){13-14}
    & Dim
    & \multicolumn{2}{c}{4}
    & \multicolumn{2}{c}{5}
    & \multicolumn{2}{c}{8}
    & 24 & 25 & 27 & 28 & 39 & 41
    \\
    \bottomrule
    \end{tabular}
\end{table}
\end{example}
\section{Conclusion remarks}\label{}
In this paper we have shown that weighted Reed-Muller codes are much
better than their reputation when defined over general point ensembles
${\mathcal{S}}=S_1\times \cdots \times S_m$. We treated in detail the
case $m=2$ and gave some results for $m>2$. It is a subject of future
studies to also  establish detailed information for the case $m>2$. We
derived two decoding algorithms that work well for different classes
of weighted Reed-Muller codes and affine variety codes
$E({\mathbb{M}},\mathcal{S})$ in general. For not too high dimensions 
these algorithms perform list decoding. For higher dimensions it is a
subject of future research to design list decoding algorithms. Using
the first algorithm in combination with some extra operations we
decoded the $[49,11,28]$ Joyner code beyond its minimum distance. It
is apparent that such an approach would work for other toric codes
coming from polytopes of the same shape.\\
This work was supported in part by Danish Natural Research Council
grant 272-07-0266. The authors gratefully acknowledge support from the
Danish National Research Foundation and the National Natural Science
Foundation of China (Grant No. xxx) for the Danish-Chinese Center for
Applications of Algebraic Geometry in Coding Theory and Cryptography. The authors would like to
thank Diego Ruano, Peter Beelen Tom H{\o}holdt, and Teo Mora for pleasant discussions. Also thanks
to L.\ Grubbe Nielsen for linguistic assistance.

\end{document}